%% file: main.tex
\documentclass[runningheads]{llncs}
\input{setup/custom_packages}

\input{setup/custom_header}

\graphicspath{{./img/}{./img/svg/}}

\begin{document}
\title{Solving Two-Player Games under Progress Assumptions\thanks{Authors are ordered randomly, denoted by
\textcircled{r}. The publicly verifiable record of the randomization is available at
\href{https://www.aeaweb.org/journals/policies/random-author-order/search?RandomAuthorsSearch\%5Bsearch\%5D=_buCtGKWs5t7}{www.aeaweb.org}.
}
}
\author{Anne-Kathrin Schmuck\inst{1} \textcircled{r}
K. S. Thejaswini\inst{2} \textcircled{r}
Irmak Sa\u{g}lam\inst{1}\textcircled{r}
Satya Prakash Nayak\inst{1}
}
\authorrunning{A-K. Schmuck \textcircled{r} K. S. Thejaswini \textcircled{r} I. Sa\u{g}lam \textcircled{r} S. P. Nayak}
\institute{Max Planck Institute for Software Systems (MPI-SWS), Kaiserslautern, Germany 
\email{\{akschmuck, isaglam, sanayak\}@mpi-sws.org}
\and Department of Computer Science, University of Warwick, UK
\email{thejaswini.raghavan.1@warwick.ac.uk}
}

\maketitle              %
\begin{abstract}
This paper considers the problem of solving infinite two-player games over finite graphs under various classes of \emph{progress assumptions} motivated by applications in cyber-physical system (CPS) design.

Formally, we consider a game graph $\gamegraph$, a temporal specification $\spec$ and a temporal assumption $\assump$, where both $\spec$ and $\assump$ are given as linear temporal logic (LTL) formulas over the vertex set of $\gamegraph$. We call the tuple $(\gamegraph,\spec,\assump)$ an \emph{augmented game} and interpret it in the classical way, i.e., winning the augmented game $(\gamegraph,\spec,\assump)$ is equivalent to winning the (standard) game $(\gamegraph,\assump\Rightarrow\spec)$. 
Given a reachability or parity game $\game = (\gamegraph,\spec)$ and some progress assumption $\assump$, this paper establishes whether solving the augmented game $\auggame = (\gamegraph,\spec,\assump)$ lies in the same complexity class as solving $\game$. 
While the answer to this question is negative for arbitrary combinations of $\spec$ and $\assump$, a positive answer results in more efficient algorithms, in particular for large game graphs.

We therefore restrict our attention to particular classes of CPS-motivated progress assumptions and establish the worst-case time complexity of the resulting augmented games. Thereby, we pave the way towards a better understanding of assumption classes that can enable the development of efficient solution algorithms in augmented two-player games.

\keywords{Synthesis  \and Graph Games \and Augmented Games \and Progress Assumptions.}
\end{abstract}

\input{sections/introduction}

\input{sections/prelim}
\input{sections/problem}

\input{sections/live-edges}
\input{sections/colive-edges}
\input{sections/live-groups}

\input{sections/progress-groups}

\bibliographystyle{splncs04}
\newpage
\bibliography{main}

\appendix
\input{sections/appendix}

\end{document}

%% file: setup/custom_packages.tex
\usepackage{mathtools,amsmath,amssymb,amsfonts}
\usepackage{stmaryrd}
\usepackage{hyperref}
\usepackage{xspace}

\usepackage{algorithm}
\usepackage[noend]{algpseudocode}
\usepackage{changepage}
\usepackage{thm-restate}
\usepackage{multirow}
\usepackage[capitalise]{cleveref}
\crefname{definition}{Def.}{Def.}
\crefname{problem}{Prob.}{Prob.}
\crefname{algorithm}{Algo.}{Algo.}
\crefname{theorem}{Thm.}{Thm.}
\crefname{lemma}{Lem.}{Lem.}
\crefname{appendix}{App.}{App.}
\crefname{line}{line}{line}
\crefname{section}{Sec.}{Sec.}
\crefname{corollary}{Cor.}{Cor.}
\crefname{proposition}{Prop.}{Prop.}
\usepackage{graphicx}
\usepackage{subcaption}
\usepackage[table,xcdraw]{xcolor}
\usepackage{longtable}
\usepackage{paralist}

\usepackage[
colorinlistoftodos, %
textwidth=\marginparwidth, %
textsize=scriptsize, %
]{todonotes}

\usepackage{mdframed}

\usepackage{pgf}
\usepackage{tikz}
\usetikzlibrary{shapes,shapes.geometric,fit,trees,decorations,arrows,automata,shadows,positioning,plotmarks,backgrounds,tikzmark}
\usetikzlibrary{calc,matrix,fit,petri,decorations.markings,decorations.pathmorphing,patterns,intersections,decorations.text}

\usepackage{graphicx}
\usepackage[utf8x]{inputenc}
\usepackage[american]{babel}
\usepackage[babel]{csquotes}

\usepackage{algorithm, algpseudocode}

\usepackage{amscd}
\usepackage{amsmath}
\usepackage{amssymb}
\usepackage{mathtools}
\usepackage{array}

\usepackage{xspace}
\usepackage{paralist}
\usepackage{xifthen}
\usepackage{url}

\usepackage{pgf}
\usepackage{tikz}
\usepackage{lmodern}

\usepackage{MnSymbol,wasysym}
\usepackage{comment}

%% file: setup/custom_header.tex
\colorlet{darkgreen}{green!80!black}
\colorlet{darkred}{red!80!black}
\usetikzlibrary{arrows, automata, shapes}
\tikzset{auto, >= stealth}
\tikzset{every edge/.append style={thick, shorten >= 1pt}}
\tikzset{initial/.style={draw, thick, <-, shorten <=1pt}}
\tikzset{player0/.style = {draw, thick, shape=circle, minimum size=5mm}}
\tikzset{player1/.style = {draw, thick, shape=rectangle, minimum size=5mm}}
\tikzset{bplayer0/.style = {draw, thick, shape=ellipse, minimum size=3mm,text width=0.65cm}}

\newcommand\hpos{1.8}
\newcommand\ypos{1.2}

\tikzset{
  bigger dots/.style={
    decoration={
      markings,
      mark=between positions 0.08 and 0.8 step 0.1 with {
        \fill[fill=white,draw=black] (0,0) circle (0.6pt); %
      },
    },
    postaction={decorate},
  },
}

\newcommand{\inodd}{\in_{\text{odd}}}
\newcommand{\ineven}{\in_{\text{even}}}
\newcommand{\abs}[1]{\left\lvert #1 \right\rvert}
\newcommand{\set}[1]{\left\lbrace #1\right\rbrace}
\newcommand{\tup}[1]{\left( #1\right)}

\newcommand{\true}{\mathsf{True}}

\newcommand{\LTLnext}{\bigcirc}
\newcommand{\LTLeventually}{\lozenge}
\newcommand{\LTLalways}{\square}

\newcommand{\lang}{\mathcal{L}}

\newcommand{\game}{\mathcal{G}}
\newcommand{\gamegraph}{\ensuremath{G}}
\newcommand{\play}{\ensuremath{\rho}}
\newcommand{\priority}{\mathbb{P}}

\newcommand{\auggame}{\ensuremath{\mathfrak{G}}}

\newcommand{\spec}{\ensuremath{\Phi}}

\newcommand{\p}[1]{\ensuremath{\text{Player}~#1}}
\newcommand{\pz}{\ensuremath{\p{0}}}
\newcommand{\po}{\ensuremath{\p{1}}}

\newcommand{\win}{\textit{Win}}
\newcommand{\winz}{\win_0}
\newcommand{\wino}{\win_1}
\newcommand{\wini}{\win_i}

\newcommand{\strat}{\ensuremath{\pi}}
\newcommand{\stratI}[1]{\ensuremath{\strat_{#1}}}
\newcommand{\stratz}{\ensuremath{\stratI{0}}}
\newcommand{\strato}{\ensuremath{\stratI{1}}}
\newcommand{\strati}{\ensuremath{\stratI{i}}}

\newcommand{\buchi}{\ifmmode B\ddot{u}chi \else B\"uchi \fi}
\newcommand{\cobuchi}{\ifmmode co\text{-}B\ddot{u}chi \else co-B\"uchi \fi}

\newcommand{\rabin}{\textit{Rabin}}
\newcommand{\parity}{\ensuremath{\textit{Parity}}}

\newcommand{\attr}[3]{\textsf{attr}\ensuremath{^{#3}\left(#1,#2\right)}}
\newcommand{\attri}[2]{\attr{#1}{#2}{i}}
\newcommand{\attro}[2]{\attr{#1}{#2}{1}}
\newcommand{\attrz}[2]{\attr{#1}{#2}{0}}
\newcommand{\Attr}[3]{\textsc{Attr}\ensuremath{^{#3}\left(#1,#2\right)}}
\newcommand{\Attri}[2]{\Attr{#1}{#2}{i}}

\newcommand{\Attrz}[2]{\Attr{#1}{#2}{0}}

\newcommand{\pre}{\textsf{pre}}

\newcommand{\livegroups}{\mathcal{H}}
\newcommand{\livegroupSingle}{H}
\newcommand{\livedges}{E^\ell}
\newcommand{\colivegroup}{E^{c}}

\newcommand{\assump}{\ensuremath{\psi}}

\newcommand{\assumplive}{\ensuremath{\assump_{\textsc{live}}}}
\newcommand{\assumpgrlive}{\ensuremath{\assump_{\textsc{glive}}}}
\newcommand{\assumpcolive}{\ensuremath{\assump_{\textsc{colive}}}}
\newcommand{\source}{\mathsf{src}}

\newcommand{\perslivegroups}{\Lambda}
\newcommand{\persedges}{\ensuremath{\mathtt C}}
\newcommand{\perssource}{\ensuremath{\mathtt S}}
\newcommand{\perstarget}{\ensuremath{\mathtt T}}
\newcommand{\assumpC}{\ensuremath{\assump_{\textsc{cont}}}}
\newcommand{\assumpPers}{\ensuremath{\assump_{\textsc{pers}}}}

\newcommand{\solve}{\textsc{Solve}}
\newcommand{\reachPers}{\textsc{AttrPers}}
\newcommand{\solveColive}{\textsc{SolveColive}}

\newcommand{\proj}{\mathrm{proj}}

\newcommand{\Ha}{\mathfrak{H}}

\newcommand{\AtrE}{\textsf{attr}^0_{\text{pers}}}
\newcommand{\AtrO}{\textsf{attr}^1_{\text{pers}}}
\newcommand{\AtrP}{\textsf{attr}^P_{\text{pers}}}
\newcommand{\AtrI}{\textsf{attr}^i_{\text{pers}}}
\newcommand{\Atr}[1]{\textsf{attr}^{#1}_{\text{pers}}}
\newcommand{\AtrPbar}{\textsf{attr}^{1-P}_{\text{pers}}}
\newcommand{\solvePers}{\textsc{SolvePers}}
\newcommand{\qsolve}{\textsc{QSolve}}

\newcommand{\wait}{\textsf{wait}}
\newcommand{\req}{\textsf{request}}
\newcommand{\grant}{\textsf{grant}}

\newcommand{\NP}{\mathrm{NP}}
\newcommand{\PTIME}{\mathrm{PTIME}}
\newcommand{\QP}{\mathrm{QP}}

\newcommand{\product}{\mathcal{P}}

\newcommand{\init}{\mathsf{init}}
\newcommand{\labeling}{\kappa}

\newcommand{\CNFname}{live CNF group\xspace}
\newcommand{\CNFnames}{live CNF groups\xspace}
\newcommand{\liveCNF}{\mathcal{F}}
\newcommand{\assumpCNFlive}{\ensuremath{\assump_{\textsc{clive}}}}

%% file: sections/introduction.tex
\section{Introduction}\label{section:introduction}

The automated and correct-by-design synthesis of control software for cyber-physical systems (CPS) has gained considerable attention in the last decades. Besides the numerous practical challenges that arise in CPS control, there are also unique theoretical challenges due to the interplay of low-level physical control loops and higher-level logical decisions. 
As a motivating example, consider a fully automated air-traffic control at an airport. While each modern airplane is equipped with feedback-controllers automatically regulating flight dynamics, the problem of assigning landing and starting spots along with non-intersecting flight corridors to each airplane is mostly solved manually by \emph{humans} these days.

The problem of designing a logical controller which automates these (higher-layer) logical decisions shares strong algorithmic similarities with the problem of \emph{reactive synthesis} from the formal methods community. In reactive synthesis, the interplay of available logical decisions with the external environment's reactions is modelled as a two-player game over a finite graph. Given such a game graph, one automatically computes a \emph{winning strategy} which takes logical decisions in response to any environment behavior such that a predefined specification always holds. %
While reactive synthesis is a mature field in computer science with rich tool support for solving games of various flavors, the application of these techniques to, e.g., the problem of air-traffic control, requires the construction of a two-player game graph, which reflects all possible interactions between the controller and the environment in an abstract manner.

Building a game graph with the `right' level of abstraction correctly modelling the interplay of low-level continuous (physical) dynamics, external interactions and logical decisions is a known severe challenge in CPS design. %
This challenge has been extensively addressed in the past decades under the term abstraction-based controller synthesis (for an overview see e.g.\ \cite{tabuada2009_book,sanfelice2020hybrid,calin_belta_book,belta2019formal}).  %
From a reactive synthesis perspective, the challenge amounts to finding 
(i) the `right' granularity of the abstract state space, i.e., the vertex set of the graph, and (ii) the `right' power of the environment player in the resulting abstract game. While the first one determines the size of the resulting graph, the second one ensures that a winning strategy for the controller does not fail to exist due to an unnecessary conservative overapproximation of environment uncertainty.

In this context, the notion of \emph{augmented} games appeared in different CPS control applications \cite{progress_groups,context-triggeredABCD,9576605,Liu21,MohajeraniMWLO21,MajumdarMSS21,NilssonOL17}. Here, the environment player is augmented with a restriction on its choice of moves, which is motivated by the physical laws governing the underlying CPS. Intuitively, these augmentations can for example abstractly model the fact that high disturbance spikes (e.g., strong wind) only occur sporadically, certain control actions (e.g., triggering a plane landing maneuver) will eventually result in a distinct system state (e.g., landing on the ground), or external logical decisions (e.g., resource allocations by a scheduler) have particular temporal poperies (e.g., are known to be fair). Here, a typical augmentation for modelling `eventual success' of a triggered control action asserts that whenever the source vertex is visited infinitely often, a dedicated transition (e.g., the one that is not-self looping at this vertex) is eventually taken by the environment player. These assumptions are also known
as \emph{strong transition fairness} \cite{baierbook,QS83,Francez} and are illustrated in \cref{fig:into-motivating}.

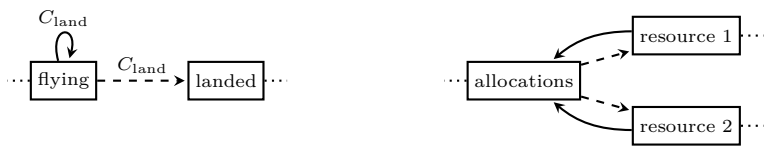
\begin{figure}[t]
    \scriptsize
    \centering
    \begin{tikzpicture}

        \node[player1] (f) at (2.1*\hpos, 0) {flying};
        \node[player1] (l) at (3.3*\hpos, 0) {landed};

        \draw[-,thick,dotted] (f.west) --+ (-0.3,0);
        \path[->] (f) edge[dashed] node[]{$C_{\text{land}}$} 
        (l) edge[loop above] node{$C_{\text{land}}$} ();
        \draw[-,thick,dotted] (l.east) --+ (0.3,0);

        \node[player1] (a) at (5.5*\hpos, 0) {allocations};
        \node[player1] (r1) at (6.7*\hpos, 0.5*\ypos) {resource 1};
        \node[player1] (r2) at (6.7*\hpos, -0.5*\ypos) {resource 2};

        \draw[-,thick,dotted] (a.west) --+ (-0.3,0);
        \path[->] (a) edge[dashed] (r1) edge[dashed] (r2);
        \path[->] (r1) edge[bend left=-20] (a);
        \path[->] (r2) edge[bend left=20] (a);
        \draw[-,thick,dotted] (r1.east) --+ (0.3,0);
        \draw[-,thick,dotted] (r2.east) --+ (0.3,0);
    \end{tikzpicture}
    \caption{Examples illustrating how physical phenomena can be abstracted by \emph{strong transition fairness assumptions} modelled by live edges (dashed) which need to be taken infinitely often if their source vertex is visited infinitely often.}\vspace{-0.5cm}%
    \label{fig:into-motivating}
\end{figure}

Formally, an \emph{augmented game} is a tuple $(\gamegraph,\spec,\assump)$, 
consisting of a finite two-player game graph $\gamegraph$, and temporal specifications $\spec$ and assumption $\assump$,
given as linear temporal logic (LTL) formulas over the vertex set of $\gamegraph$. An augmented game is typically interpreted in the classical way, i.e., winning the augmented game $(\gamegraph,\spec,\assump)$ is equivalent to winning the (standard) game $(\gamegraph,\assump\Rightarrow\spec)$ with modified specification $\Phi':=\assump\Rightarrow\spec$. 

As augmented games can always be solved via their \enquote{implication-form} $(\gamegraph,\assump\Rightarrow\spec)$, it is not immediately obvious that augmentations can fundamentally change the given reactive synthesis problem, in particular their complexity class. 
However, game graphs resulting from CPS abstractions tend to be very large while specifications are rather simple, i.e., typically induced by `reach-while-avoid' problems. Given for example a reachability game augmented with strong transition fairness assumptions, results in a Rabin game (with one Rabin pair per fair transition) when re-written in implication form, hence moving from a PTIME problem to solve the original (non-augmented) game to an NP-complete problem for its augmented version. Given a very large graph with augmentations on almost every vertex, this can very quickly lead to computational intractability.
On the other hand, augmentations are typically very structured and local (as e.g., strong transition fairness). Using these features, it was recently shown by Banerjee et al.~\cite{banerjee2022fast} that strong transition fairness assumptions can actually be handled `for free' in synthesis -- i.e., solving reachability games augmented with strong transition fairness assumptions can be solved in PTIME.

\begin{table}[b]
\centering
\begin{tabular}{|c||c|c|}
\hline
\textbf{Augmentation Type}                                         & \textbf{Reachability Games} & \textbf{Parity Games} \\ \hline\hline
no augmentations   & $\PTIME$~(\cref{thm:reach-ptime})       & $\QP$~(\cref{thm:parity-qp})          \\ \hline
live edges~(\cref{def:live-edges})      & $\PTIME$~(\cref{thm:live-edges})      & $\QP$~(\cref{thm:live-edges})       \\ \hline
co-live edges~(\cref{def:colive})    & $\PTIME$~(\cref{thm:co-live})       & $\QP$~(\cref{thm:co-live})           \\ \hline
live groups~(\cref{def:live-groups})     & \hspace{0.2cm} $\NP$-complete~(\cref{thm:NP-completeness-augmentedreachparity}) \hspace{0.2cm}& \hspace{0.2cm}$\NP$-complete~(\cref{thm:NP-completeness-augmentedreachparity}) \hspace{0.2cm}\\ \hline
\multirow{2}{*}{\shortstack{singleton-source\\ live groups~(\cref{section:CNF-edges})}} & \multirow{2}{*}{$\PTIME$~(\cref{thm:singleton-live-groups})} & \multirow{2}{*}{$\QP$~(\cref{thm:singleton-live-groups})} \\ 
& & \\ \hline
persistent live groups~(\cref{def:perslivegroups}) & $\PTIME$~(\cref{corollary:progress-reach-ptime} )      & $\QP$~(\cref{theorem:pers-solve_augmented_parity})          \\ \hline
\end{tabular}
\caption{Summary of complexity classes for solving the augmented games discussed in this paper.}
\label{table:complexity}
\end{table}

Motivated by this result, this paper investigates which other classes of augmentations can be \enquote{handled for free} in \emph{reachability and parity games}, as summarized in \cref{table:complexity}. In particular, we show that parity games (resp. reachability games) augmented with live edges (strong transition fairness), co-live edges (strong transition \emph{co-fairness}), singleton-source live groups or persistent live groups can be solved in quasi-polynomial time (resp. polynomial time), which largely extends the known class of efficiently solvable augmented games. In addition, we also prove NP-completeness of a slightly more general, but most frequent class of assumptions, called live groups (group transition fairness) assumptions. Unfortunately, this result shows that solving augmented games resulting from the product of different components augmented with strong transition fairness is an NP-hard problem, even for simple reachability objectives.

We note that most of the progress assumptions we consider in this paper solely restrict the moves of the environment player, and are hence non-falsifiable by the system player. This implies that, in most cases, the system player cannot vacuously win the implication-version $(\gamegraph,\assump\Rightarrow\spec)$ of an augmented game $(\gamegraph,\spec,\assump)$ by enforcing $\neg \assump$ (instead of $\spec$). This distinguishes our results from the large body of work investigating how winning strategies should be synthesized in augmented games such that (general LTL) assumptions are \enquote{handled correctly}, as e.g.\ in \cite{KleinPnueli-2010,DBPU10,chatterjee2010obliging,BCGHHJKK14,BloemEhlersKoenighofer_2015,AssumptionsInSynthesis,majumdar2019environmentally}. 
In addition, we are also not concerned with the problem of \emph{computing} assumptions that render non-realizable synthesis problems realizable, as e.g.\ in \cite{ChatterjeeHenzingerJobstmann_2008,maoz2019symbolic,permissiveAssumptions}. In this paper, assumptions are assumed to be \emph{given} and resulting from an abstraction process of existing component dynamics.

%% file: sections/prelim.tex
\section{Preliminaries}\label{section:prelim}
\noindent\textbf{Notation.}
Given two real numbers $a,b\in\mathbb{R}$ with $a<b$, we use $[a;b]$ to denote the set $\set{n\in\mathbb{N} \mid a\leq n\leq b}$ of all integers between $a$ and $b$.
For any given set $[a;b]$, we write $i\ineven [a;b]$ and $i\inodd [a;b]$ as shorthand for $i\in [a;b]\cap \set{0,2,4,\ldots}$ and $i\in [a;b]\cap \set{1,3,5,\ldots}$ respectively.
Given two sets $A$ and $B$, a relation $R\subseteq A\times B$, and an element $a\in A$, we write $R(a)$ to denote the set $\set{b\in B\mid (a,b)\in R}$.

\smallskip
\noindent\textbf{Languages.}
Let $\Sigma$ be a finite alphabet. $\Sigma^*$ and $\Sigma^\omega$ denote, respectively, the set of finite and infinite words over $\Sigma$, and $\Sigma^\infty$ is equal to $\Sigma^*\cup \Sigma^\omega$.
Given two words $u\in \Sigma^*$ and $v\in \Sigma^\infty$, their concatenation is written as the word $uv$.

\smallskip
\noindent\textbf{Game Graphs.}
A \emph{game graph} is a tuple $\gamegraph= \tup{V= V_0\uplus  V_1,E}$ where $(V,E)$ is a finite directed graph with \emph{vertices} $ V $ partitioned into two sets and \emph{edges} $ E \subseteq V\times V$.
For vertices $u,v\in V$, we write $(u,v)$ or $u\to v$ to denote an edge from $u$ to $v$.
For a set $E'\subseteq E$ of edges, we write $\source(E') = \{u \mid (u,v)\in E'\}$ to denote the sources of the edges in $E'$.
Furthermore, for each $i\in\{0,1\}$, we write $E_i$ to denote the set of edges originating from $V_i$.
A \emph{play} $\play$ originating at a vertex $v_0$ is a finite or infinite sequence of vertices where consecutive ones are connected by edges, denoted by $v_0v_1\ldots \in V^\infty$ or $v_0\to v_1\to \cdots\in V^\infty$.
We sometimes fix a designated initial vertex $v_0$ and disregard the subgraph not reachable from $v_0$. 

\smallskip
\noindent\textbf{Specifications/Objectives.}
Given a game graph $\gamegraph$, we consider specifications/objectives specified using a formula $\spec$ in \emph{linear temporal logic} (LTL) over the vertex set $V$, that is, we consider LTL formulas whose atomic propositions are sets of vertices $V$. 
In this case the set of desired infinite plays is given by the semantics of $\spec$ which is an $\omega$-regular language $\lang(\spec)\subseteq V^\omega$. 
Every game graph with an arbitrary $\omega$-regular set of desired infinite plays can be reduced to a game graph (possibly with a different set of vertices) with an LTL specification, as above\footnote{This is because every $\omega$-regular specification can be reduced to a parity specification~\cite{gamesbook}, which can be written as an LTL specification over $V$ as given later.}.
The standard definitions of $\omega$-regular languages and LTL are omitted for brevity and can be found in standard textbooks \cite{baierbook}. To simplify notation we use $ e=(u,v) $ in LTL formulas as syntactic sugar for $ u\wedge \bigcirc v $, with $ \bigcirc$ as the LTL \emph{next} operator. We further use a set of edges $E' = \set{e_i}_{i\in [0;k]}$ as an atomic proposition to denote $\bigvee_{i\in [0;k]} e_i$. 

\smallskip
\noindent\textbf{Games and Strategies.} A \emph{two-player (turn-based) game} is a pair $\game=\tup{\gamegraph,\spec}$ where $G $ is a game graph and 
$ \spec $ is a \emph{specification} over $\gamegraph$.
A \emph{strategy} of $\p{i},~i\in\{0,1\}$, is a function $\strati\colon  V^* V_i\to  V$ such that for every $\play v \in  V^* V_i$ holds that $\strati(\play v)\in  E(v)$. %
Furthermore, a strategy $\strati$ is \emph{memoryless}/\emph{positional} if $\strati(\play v) = \strat(v)$ for every $\play v\in V^*V_i$. We write such memoryless strategies as functions of the form $\strati\colon V_i\rightarrow V$.
Given a strategy $\strati$, we say that the play $\play=v_0v_1\ldots$ is \emph{compliant} with $\strati$ if $v_{k-1}\in  V_i$ implies $v_{k} = \strati(v_0\ldots v_{k-1})$ for all $k$.
We refer to a play compliant with $\strati$ and a play compliant with both $\stratz$ and $\strato$ as a \emph{$ \strati $-play} and a \emph{$ \stratz\strato $-play}, respectively. 
We collect all plays originating in a set $ S $ and compliant with $\strati$, (and compliant with both $\stratz$ and $\strato$) in the sets $\lang(S,\strati)$ (and $\lang(S,\stratz\strato)$, respectively). When $ S=V $, we drop the mention of the set in the previous notation, and when $ S $ is singleton $ \{v\} $, we simply write $ \lang(v,\strati) $ (and $ \lang(v,\stratz\strato) $, respectively). 

\smallskip
\noindent\textbf{Winning.}
Given a game $\game=(\gamegraph,\spec)$, a play $ \play $ in $ \game $ is \emph{winning for $ \p{0} $}, if $ \play\in\lang(\spec) $, and it is winning for $ \p{1} $, otherwise. A strategy $\strati$ for $\p{i}$ is \emph{winning from a vertex $ v\in V $} if all plays compliant with $ \strati $ and originating from $ v $ are winning for $ \p{i} $. 
We say that a vertex $v\in V$ is  \emph{winning for $\p{i}$}, if there exists a winning strategy $\strati$ from $ v $. %
We collect all winning vertices of $\p{i}$ in the \emph{$\p{i}$ winning region} $\wini\subseteq V$.
A (uniform) \emph{winning strategy} for $\p{i}$ is a strategy that is winning for $\p{i}$ from all vertices in $\wini$.
We always interpret winning w.r.t.\ $\p{0}$ if not stated otherwise.
Furthermore, in games with a designated initial vertex $v_0$, when we discuss the winner of the game we refer to the winner of $v_0$.

\smallskip
\noindent\textbf{Reachability Games.} 
A \emph{reachability game} is a game $\game=(\gamegraph,\spec)$ with reachability objective $\spec = \LTLeventually T$ for some set $T$ of vertices.
A play is winning for $ \pz $ in such a game if it visits any vertex in $T$.
It is well-known that reachability games can be solved in linear time in the size of the game graph.
\begin{lemma}[\cite{Thomas97}]\label{thm:reach-ptime}
    The problem of solving reachability games lies in $\PTIME$.
\end{lemma}

\smallskip
\noindent\textbf{Parity Games.} 
A \emph{parity game} is a game $\game=(\gamegraph,\spec)$ with parity objective $\spec = \parity(\priority)$, s.t.\
 $\textstyle\parity(\priority)\coloneqq \bigwedge_{i\inodd [0;d]} \left(\LTLalways\LTLeventually P_i \implies \bigvee_{j\ineven [i+1;d]} \LTLalways\LTLeventually P_j\right)$, 
with $ P_i=\{v\in V\mid \priority(v)=i \} $ for some priority function $ \priority: V\rightarrow [0;d] $ that assigns each vertex a priority. A play is winning (for $ \pz $) in such a game if the maximum of priorities seen infinitely often is even. 
There are several quasi-polynomial algorithms to solve such parity games giving us the following result.
\begin{lemma}[\cite{CaludeJKL017,quasiZielonka,LehtinenPSW22}]\label{thm:parity-qp}
    The problem of solving parity games lies in $\QP$.
\end{lemma}

\smallskip
\noindent\textbf{Rabin Games.} 
A \emph{Rabin game} is a game $\game=(\gamegraph,\spec)$ with Rabin objective $\spec = \rabin(\Omega)$ defined by a set of \emph{Rabin pairs} $\Omega = \{(F_i,R_i)\subseteq V\times V\mid i\in[1;k]\}$ s.t.\
 $\textstyle\rabin(\Omega)\coloneqq \bigvee_{i\in [1;k]} \left(\LTLalways\LTLeventually F_i \wedge \neg\LTLalways\LTLeventually R_i\right)$.
A play is winning for $ \pz $ in such a game if there exists an $i\in [1;k]$ such that the plays visits the set $F_i$ infinitely often and $R_i$ only finitely often.
The dual of such an objective, i.e. $\neg\rabin(\Omega)$, is known as a \emph{Streett} objective.
It is well-known that solving Rabin games for $\pz$ is an $\NP$-complete problem whereas solving it for $\po$, which corresponds to solving Streett games for $\pz$, is $\mathrm{co}\text{-}\NP$-complete.
\begin{lemma}[\cite{EmersonJ91,EmersonJ99,Thomas97}]\label{thm:rabin-np-complete}
    The problem of solving Rabin games for $\pz$ is $\NP$-complete and for $\po$ it is $\mathrm{co}\text{-}\NP$-complete.
\end{lemma}

\smallskip
\noindent\textbf{Computing attractors and winning regions.}
For a set $T$ of vertices, $\pre(T)\subseteq V$ is the set of vertices from which there is an edge to $T$.
Furthermore, the attractor function $\Attri{\gamegraph}{T}$ solves the (non-augmented) reachability game $(G,\LTLeventually T)$ for $\p{i}$, i.e., it returns the attractor set, i.e., winning region, $A:=\attri{\gamegraph}{T}\subseteq V$ and a memoryless attractor strategy, i.e., winning strategy, $\strat_A$ of $\p{i}$. Intuitively, $A$ collects all vertices from which $\p{i}$ has a strategy (i.e., $\strat_A$) to force every play starting in $A$ to visit $T$ in a finite number of steps.
Moreover, the function $\solve(\gamegraph,\spec)$ returns the winning region and a winning strategy (that is memoryless if possible) in a game $(\gamegraph,\spec)$.
Both the functions $\textsc{Attr}$ and $\solve$ (for a parity objective) solve classical synthesis problems with standard algorithms (see, e.g.\ \cite{gamesbook}).

%% file: sections/problem.tex
\section{Problem Statement}

In this paper we consider the problem of efficiently solving reachability and parity games $(\gamegraph,\spec)$ which are augmented with assumptions, as formalized next.

\begin{definition}
  An \emph{augmented game} $\auggame=(\gamegraph,\spec,\assump)$ is a game $(\gamegraph,\spec)$ augmented with an \emph{assumption} $\assump$ given as an LTL formula over the vertex set $V$ of $\gamegraph$. $\auggame$ is interpreted as
the (non-augmented) game $\game = (\gamegraph, \assump\Rightarrow\spec)$,
i.e., $\p{i}$ wins in $\auggame$ if s/he wins in $\game$.
  An augmented game with reachability (resp. parity) objective~$\spec$ is called an augmented reachability (resp. parity) game.
\end{definition}

\begin{problem}\label{problem:main}
  Given a class of assumptions, does the problem of solving the respective augmented reachability (resp. parity) game lie in $\PTIME$ (resp. $\QP$)?
\end{problem}
It is easy to see that the answer to \cref{problem:main} is negative for arbitrary combinations of $\spec$ and $\assump$. A simple example is a reachability game $\auggame$ augmented with a Streett assumption $\assump$.
Solving such games is an $\NP$-complete problem, as it is equivalent to solving Rabin games, due to the duality of Rabin and Streett objectives (\cref{thm:rabin-np-complete}), while solving (non-augmented) reachability games lies in $\PTIME$ (\cref{thm:reach-ptime}).
We therefore restrict our attention to particular classes of \emph{progress} assumptions, explain their relevance in the context of CPS design and establish the worst-case time complexity of the resulting augmented games. Our results are summarized in \cref{table:complexity}.
On a higher level, the goal of this paper is to pave the way towards a more comprehensive understanding of assumption classes that allow for efficient solution algorithms of augmented games.

%% file: sections/live-edges.tex
\section{Strong Transition Fairness Assumptions}\label{section:live-edges}
Fairness assumptions have proven to be useful to prevent reactive synthesis problems to become unrealizable for `uninteresting reasons'. For example, synthesizing a mutual exclusion protocol might fail because a process that entered the critical section might decide to never leave it, allowing no other process which requested access to enter. This can be circumvented by a fairness constraint which asserts that every process will eventually leave the critical section again (which must be ensured by its local implementation). This reasoning analogously holds for CPS, where a triggered landing maneuver will lead to successfully landing (assuming the low-level controller being correct), or attempting to grasp an object will eventually succeed (assuming the robot to be programmed well). 

Depending on the game graph used to model the overall synthesis problem, the outlined `fair progress' can be captured by a local fairness notion called \emph{strong transition fairness} \cite{QS83,Francez,baierbook}, as used for resource management \cite{CAFMR13}, path following \cite{MMSS2021,DIRS18,NOL17} or robot manipulators \cite{AGR20}.
This assumption class is given by a set of \emph{live edges} and requires that whenever the source vertex of a live edge is visited infinitely often along a play, the edge itself is traversed infinitely often along the play as well (see \cref{fig:into-motivating} for different illustrations). 
Whenever strong transition fairness is used as an \emph{assumption}, it is solely restricting the environment player in the resulting game, as formalized next.

\begin{definition}\label{def:live-edges}
Given a game graph $\gamegraph = (V,E)$, \emph{strong transition fairness assumptions} $\assumplive(\livedges)$ are represented by a set $\livedges \subseteq E_1$ of \emph{live edges} and captured by the LTL formula
\begin{equation}\label{eq:assumpLive}
   \textstyle\assumplive(\livedges) \coloneqq \bigwedge_{e = (u,v)\in \livedges} \left(\LTLalways\LTLeventually u \Rightarrow \LTLalways\LTLeventually e\right).
\end{equation}
We call $\auggame = (\gamegraph,\spec,\assumplive(\livedges))$ a game augmented with live edges $\livedges$.
\end{definition}

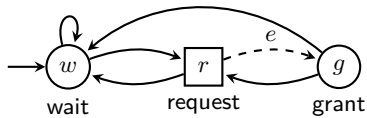
\begin{figure}[t]
    \centering
    \begin{tikzpicture}
        \node[player0,label={below:\wait}] (w) at (0, 0) {$w$};
        \node[player1,label={below:\req}] (r) at (\hpos, 0) {$r$};
        \node[player0,label={below:\grant}] (g) at (2*\hpos, 0) {$g$};

        \draw[<-,thick] (w.west) --+ (-0.5,0);
        \path[->] (w) edge[bend left=20] (r) edge[loop above] ();
        \path[->] (r) edge[bend left=20] (w);
        \path[->,dashed] (r) edge[bend left=20] node{$e$} (g);
        \path[->] (g) edge[bend left=20] (r) edge[bend left=-40] (w);
    \end{tikzpicture}
    \caption{A reachability game with $\po$ (squares) vertices, $\pz$ (circles) vertices, and specification $\spec = \LTLeventually g$ augmented with (dashed) live edge $\livedges = \{e\}$.}\label{fig:live-edges}
\end{figure}
\begin{example}\label{example:live-edges}
    A simple example of a game augmented with live edges is shown in \cref{fig:live-edges}. Intuitively, in this game graph, from every $\pz$ vertex, she can choose to either wait (by going to vertex~$w$) or request (by going to vertex~$r$), and from $\po$ vertex $r$, he can either grant the request (by going to vertex~$g$) or make $\pz$ wait (by going to vertex~$w$).
    Furthermore, starting from vertex $w$, $\pz$'s objective, i.e., $\spec = \LTLeventually g$, is to finally get her request granted by visiting vertex~$g$.
    It is easy to see that without any assumption, $\pz$ does not have a winning strategy from $w$, as $\po$ can always choose to make $\pz$ wait.
    Now consider the assumption $\assumplive(\livedges)$ for live edges $\livedges = \{e\}$, which says that if $\pz$ requests infinitely often, then $\po$ grants infinitely often.
    Under this assumption, now, $\pz$ can win by requesting again and again until $\po$ grants. Hence, $\pz$ has a winning strategy in this augmented game with assumption $\assumplive(\livedges)$. 
\end{example}

It turns out, that for this assumption class the answer to  \cref{problem:main} is positive.

\begin{theorem}\label{thm:live-edges}
    Reachability (resp. parity) games augmented with live edges can be solved in $\PTIME$ (resp. $\QP$).
\end{theorem}

\cref{thm:live-edges} is actually a corollary of known results, due to the following observations.
First, it was recently shown by Banerjee et al.~\cite{banerjee2022fast} that both reachability and parity
games augmented with live edges can be solved using a symbolic fixed-point
algorithm. Their algorithm encodes winning regions of such games as the solutions of fixpoint equations, and shows that one can solve the equations obtained with almost the same computational (worst case) complexity as the standard fixed-point algorithm for the corresponding \emph{non-augmented} games. 

The winning region of parity games augmented with live edges can be encoded further as the solution of \emph{nested fixpoint equations}, where the nesting depth of the obtained equations depend only on the number of priorities~\cite[Sec.3.4]{banerjee2022fast}. 
Following the quasi-polynomial algorithms for parity games, several results show that these techniques can be extended to solve arbitrary nested fixpoint equations in quasi-polynomial time~\cite{ANP21,HS21,JMT22}. 
As a corollary of these results, the existence of quasi-polynomial time algorithms for such augmented parity games follows.
Furthermore, as reachability games augmented with live edges can be encoded by a nested fixpoint equation whose nesting depth is fixed, they can be solved in polynomial time.

As a byproduct of this result, qualitative winning in stochastic parity games can also be decided in similar time complexity by reducing them to augmented games with live edges (see \cite[Sec.5]{banerjee2022fast} for the formal reduction). 
Stochastic two-player games (also known as $2\frac{1}{2}$-player games) generalize two-player graph games with an additional category of ``random'' vertices: whenever the game reaches a random vertex, a random process picks one of the outgoing edges (uniformly at random, w.l.o.g.). The qualitative winning problem asks whether a vertex of the game graph is almost surely (with probability $1$) winning for $\pz$.

\begin{corollary}
   The qualitative winning problem in stochastic parity games can be solved in $\QP$.
\end{corollary}

%% file: sections/colive-edges.tex
\section{Strong Transition Co-Fairness Assumptions}\label{section:colive-edges}
Even though fairness assumptions have arguably received much more attention in the reactive synthesis community, their dual -- \emph{co-fairness assumptions} -- also provide a simple and local abstraction of (non)-progress behaviour. An examples are capacity or energy restrictions, which prevent a robot manipulator to place infinitely many pieces into a buffer before it is emptied. %
Formally, strong transition \emph{co-fairness} requires a particular set of \emph{co-live edges} to be taken only finitely often in every play, as formalized next.

\begin{definition}\label{def:colive}
Given a game graph $\gamegraph = (V,E)$, \emph{strong transition co-fairness assumptions} are represented by a set $\colivegroup \subseteq E_1$ of \emph{co-live edges} and captured by the LTL formula
\begin{equation}\label{eq:assumpColive}
   \textstyle\assumpcolive(\colivegroup) \coloneqq \bigwedge_{e\in \colivegroup} \neg\LTLalways\LTLeventually e.
\end{equation}
We call $\auggame = (\gamegraph,\spec,\assumpcolive(\colivegroup))$ a game augmented with co-live edges $\colivegroup$. 
\end{definition}

The concept of co-live edges along with their induced assumption on the environment player was recently introduced by Anand et al.~\cite{permissiveAssumptions}. In their work, the problem of computing adequately permissive assumptions which render a given non-realizable (non-augmented) synthesis problem realizable, was studied. In contrast, we are interested in the problem where such assumptions are \emph{given} (as a result of the modelling process) and need to be utilized in synthesis. %

\begin{algorithm}[b]
    \caption{$ \solveColive(\gamegraph, \spec, \colivegroup) $}\label{alg:colive}
    \begin{algorithmic}[1]
        \Require Augmented parity game $\auggame = (\gamegraph = (V,E),\spec,\assumpcolive(\colivegroup))$
        \Ensure Winning region and winning strategy in the augmented game $\auggame$  
        \State $\gamegraph' \gets (V,E\setminus\colivegroup)$
        \State $\winz',\stratz' \gets \solve({\gamegraph',\spec})$; $\wino'\gets V\setminus \winz'$\label{alg:colive:solve1}
        \State $A \gets \attro{\gamegraph}{\wino'}$; $B\gets V\setminus A$\label{alg:colive:attractor}
        \If{$B = \emptyset$ or $B = V$}\label{alg:colive:if}
             \Return $B,\stratz'$\label{alg:colive:return1}
        \Else
            ~\Return $\solveColive(\game|_B)$\label{alg:colive:return2}
        \EndIf
    \end{algorithmic}
\end{algorithm}

\cref{alg:colive} provides an algorithm to solve parity games augmented with co-live edges. The idea of the algorithm is to first solve the parity game w.r.t. the game graph obtained by removing the co-live edges (\cref{alg:colive:solve1}) to get the current winning regions $(\winz',\wino')$. As $\po$ can use the co-live edges finitely many times to reach his winning region, we compute the $\po$ attractor set $A$ for $\po$'s current winning region $\wino'$ (\cref{alg:colive:attractor}), which gives us a subset of his complete winning region. If $A\subseteq V$ is a non-trivial (i.e., non-empty and strict subset), then we re-solve the game on the game restricted to its complement $B = V\setminus A$ (\cref{alg:colive:return2}). This is formalized in the following theorem and proven in \cref{appendix:co-live}. 

\begin{restatable}{theorem}{restatecolive}\label{thm:colive}
Given an augmented parity game  $\auggame = (\gamegraph,\spec,\assumpcolive(\colivegroup))$ with game graph $\gamegraph = (V,E)$ and co-live edges $\colivegroup$, the algorithm $\solveColive(\gamegraph, \spec, \colivegroup)$ returns the winning region and a memoryless winning strategy in game $\auggame$.
\end{restatable}

As in each iteration, $\solveColive$ (\cref{alg:colive}) restricts the game graph to a smaller vertex set, the algorithm terminates within $\abs{V}$ iterations. 
As reachability games can be reduced to simple parity games~\cite{gamesbook},
and since non-augmented parity (resp. reachability) games can be solved in quasi-polynomial (resp. polynomial) time (\cref{thm:parity-qp,thm:reach-ptime}), \cref{alg:colive} can solve parity (resp. reachability) games augmented with co-live edges in quasi-polynomial (resp. polynomial) time.
\begin{corollary}\label{thm:co-live}
    Reachability (resp. parity) games augmented with strong transition co-fairness assumptions can be solved in $\PTIME$ (resp. $\QP$).
\end{corollary}

%% file: sections/live-groups.tex
\section{Group Transition Fairness Assumptions}\label{section:live-groups}
In the previous sections, we saw that augmenting reachability or parity games with live or co-live edges allows for solution algorithms with the same worst-case time complexity as for the respective \emph{non-augmented} games. However, in many scenarios, live and co-live edges are not expressive enough to capture the intended progress assumption, as illustrated in the following example.

\begin{figure}[t]
    \centering
    \begin{tikzpicture}
        \node[player0] (w) at (0, 0) {$w_1$};
        \node[player1] (r) at (\hpos, 0) {$r_1$};
        \node[player0] (g) at (2*\hpos, 0) {$g_1$};
        \node[player1] (r2) at (3*\hpos, 0.5*\ypos) {$r_2$};
        \node[player0] (w2) at (3*\hpos, -0.5*\ypos) {$w_2$};
        \node[player0] (g2) at (4*\hpos, 0) {$g_2$};

        \draw[<-,thick] (w.west) --+ (-0.5,0);
        \path[->] (w) edge[bend left=20] (r) edge[loop above] ();
        \path[->] (r) edge[bend left=20] (w);
        \path[->,dashed] (r) edge node{$e_1$} (g);
        \path[->] (g) edge (r2) edge (w2);
        \path[->] (r2) edge[dashed] node{$e_2$} (g2) edge[bend left=-30] (w);
        \path[->] (w2) edge[bend left=20] (r) edge[bend left=30] (w);
    \end{tikzpicture}
    \caption{Augmented reachability game $(\gamegraph_2,\spec_2,\assumpgrlive(\livegroups))$ with $\po$ vertices (squares), $\pz$ vertices (circles), specification $\spec_2 = \LTLeventually g_2$ and live group $\livegroups = \{\{e_1,e_2\}\}$. This game is equivalent to the augmented game $(\gamegraph,\spec_1,\assumplive(\livedges))$, with $\gamegraph$, $\livedges = \{e\}$ as depicted in \cref{fig:live-edges} and $\spec_1:=\LTLeventually (g \wedge \LTLnext\LTLnext g)$.}\vspace{-0.5cm}\label{fig:live-groups}
\end{figure}
\begin{example}\label{example:live-groups}
    Consider again the game graph $\gamegraph$ as shown in \cref{fig:live-edges} with a new specification $\spec_1 =  \LTLeventually (g \wedge \LTLnext\LTLnext g)$. 
    Intuitively, $\pz$ wants to ensure that at some point she gets two grants consecutively.
    However, under the live edge assumption $\assumplive(\livedges)$ with $\livedges = \{e\}$, $\pz$ does not have a winning strategy. This live edge only ensures infinitely many grants, but does not ensure consecutive grants. 
    As the specification is not a reachability/parity objective, the standard way to solve the game $(\gamegraph,\spec_1)$ is to translate $\spec_1$ into a parity game $\game'$ and then take the product of $\game'$ with $\gamegraph$. In the given example, the result turns out to be the reachability game $(\gamegraph_2,\spec_2)$ with $\spec_2 = \LTLeventually g_2$, as depicted \cref{fig:live-groups}. 
    In this product game graph, live edge $e$ of $\gamegraph$ has been split into two copies, i.e., edge $e_1$ and $e_2$.  
    If we consider the live edge assumption $\assumplive(\livedges_2)$ with $\livedges_2 = \{e_1,e_2\}$, then $\pz$ can satisfy $\spec_2$ by requesting every time -- this will force $\po$ to take edge $e_1$ infinitely often and then finally take edge $e_2$.
    However, $\pz$ should actually not have a winning strategy as she does not have one in game $(\gamegraph,\spec_1)$.
    This tells us that the live edge $e$ in $G$ does not translate to live edges $\{e_1,e_2\}$ in the product game.
    In particular, in order to satisfy the fairness assumption inherited from the original game, it is sufficient for $\po$ to take only one of the edges in $\{e_1,e_2\}$ infinitely often if their sources are visited infinitely often.
    This disjunction over live edges can be expressed by a \emph{live group} $\livegroupSingle = \{e_1,e_2\}$, as formalized next. With this, the augmented game $(\gamegraph_2,\spec_2,\assumpgrlive(\{\livegroupSingle\}))$ becomes equivalent to $(\gamegraph,\spec_1,\assumplive(\livedges))$, and is therefore again unrealizable.
\end{example}

\subsection{Live Group Assumptions}\label{subsec:live-group-def}

\emph{Live group} assumptions, a generalization of live edge assumptions, are defined by a finite set $\livegroups$ of edge groups $\livegroupSingle$ in a game. The assumptions require that for each live group $\livegroupSingle \in \livegroups$, if at least one source vertex in $\livegroupSingle$ is visited infinitely often along a play, at least one of the edges in $\livegroupSingle$ is traversed infinitely often as well.

\begin{definition}\label{def:live-groups}
Given a game graph $\gamegraph = (V,E)$, the assumptions represented by a set $\livegroups$ of live groups $\livegroupSingle$ are captured by the LTL formula
\begin{equation}\label{eq:assumpLiveGroup}
    \textstyle\assumpgrlive(\livegroups) \coloneqq \bigwedge_{\livegroupSingle\in \livegroups} \left(\LTLalways\LTLeventually\source(\livegroupSingle) \Rightarrow \LTLalways\LTLeventually \livegroupSingle\right).
\end{equation}
We call $\auggame = (\gamegraph,\spec,\assumpgrlive(\livegroups))$ a game augmented with live groups $\livegroups$.
\end{definition}

\noindent Unfortunately, it turns out that for such games the answer to \cref{problem:main} is negative.

\begin{restatable}{theorem}{restatenpcompletenesslivegrps}\label{thm:NP-completeness-augmentedreachparity}
     Solving reachability and parity games augmented with live groups is $\NP$-complete.%
\end{restatable}

In order to prove this theorem, we will introduce \cref{lem:NPinclusion-augmentedparity} and \cref{lem:NPhardness-augmentedreachability}. %
Since reachability specifications give one of the easiest infinite games, this result can be perceived as a negative result for all meaningful games augmented with live group assumptions. However, 
the $\NP$-completeness result also holding for parity implies that going from reachability to parity does not add up to the complexity.

\begin{restatable}{lemma}{restatelivegrouplemma}\label{lem:NPinclusion-augmentedparity}
    Parity games augmented with live groups can be solved in $\NP$.
\end{restatable}
\noindent \textit{Proof sketch [Full proof in \cref{appendix:lem:NPinclusion-augmentedparity}].} The intuition is that we can encode both the parity conditions and the live group conditions as Rabin pairs.
For an augmented parity game $\auggame =(\gamegraph = (V,E),\spec = \parity(\priority), \assumpgrlive(\livegroups) )$, 
an equivalent Rabin game is $\game' = ((V', E'), \spec' = \rabin(\Omega_1 \cup \Omega_2))$ with $V' = V \uplus E$ and $E' = \{ (u,e), (e,v) \mid (u,v) \in E\}$ where
the Rabin pairs in $\Omega_1 = \{(\source(\livegroupSingle),\livegroupSingle) \mid \livegroupSingle\in\livegroups\}$ represent the live group conditions and
the ones in $\Omega_2 = \{(P_{2i},\cup_{j>2i} P_j)\mid 0\leq 2i\leq d\}$ where $ P_i=\{q\in Q\mid \priority(q)=i \}$ represent the parity conditions.\qed

\begin{restatable}{lemma}{restatelivegrouplemmab}\label{lem:NPhardness-augmentedreachability}
    Solving reachability games augmented with live groups is $\NP$-hard.
\end{restatable}
\noindent \textit{Proof sketch [Full proof in \cref{appendix:lem:NPinclusion-augmentedparity}].} 
    We will give a polynomial-time reduction from the 3-SAT problem to reachability games augmented with live groups.
    Towards this goal, we consider the 3-SAT instance 
    $\varphi = C_1\wedge C_2 \wedge\cdots \wedge C_k$ where $C_i = (y_{1i} \vee y_{2i} \vee y_{3i})$ for $i \in [1;k]$
    and each $y_{1i},y_{2i},y_{3i}$ is a literal from the set $\{x_1, x_2, \ldots, x_m, \neg x_1, \neg x_2, \ldots, \neg x_m\}$.
    We construct the augmented reachability game $\auggame^\varphi =(\gamegraph = (V,E),\spec = \LTLeventually\smiley{}, \assumpgrlive(\livegroups))$ 
    with vertex set 
    \begin{equation*} V = \{v_0\} \cup \{\tiny{C_i} \mid i \in [1; k]\} \cup \{y, y' \mid y \in \{x_i, \neg x_i\} \quad\text{\large{for}}\quad  i \in [1;m]\} \cup \{\large{\smiley{}}\} 
    \end{equation*}
    where $\p{0}$ vertices consists of $\{\tiny{C_i} \mid i \in [1;m] \}$ and $\large{\smiley{}}$; edge set
    \footnote{The edges $x_i \to x'_i$ mainly serve illustrative purposes, and the live outgoing edge of $x'_i$ can actually be attributed directly to $x_i$.
    Further, distributing live edges to separate vertices underscores the result's validity for live edges with disjoint sources.  
    }     
    \begin{align*}
        E = &\{ (v_0, C_i) \mid i \in [1;k]\} \cup \{(y, \large{\smiley{}}), (y, y') \mid y \in \{x_i, \neg x_i\} \text{ for } i \in [1;m]\} \cup \\
        &\{(C_j, y) \mid j \in [1;m], y \in \{x_i, \neg x_i\} \text{ for } i \in [1;m] \text{ and } y \text{ is in clause }C_j\} \cup \{(\large{\smiley{}}, \large{\smiley{}})\};
    \end{align*}
    and live groups $\livegroups = \{ H_i^1, H_i^2 \mid i \in [1;m]\}$ where $H_i^1 := \{(x_i, \large{\smiley{}}), (\neg x'_i, v_0)\}$ and $ H_i^2 := \{(\neg x_i, \large{\smiley{}}), (x'_i, v_0)\}$.
    The game $\auggame^\varphi$ for the 3-SAT formula $\varphi = C_1 \wedge C_2 \wedge C_3$ with $C_1 = (x_1 \vee x_2 \vee \neg x_3)$, $C_2 = (\neg x_1 \vee x_2 \vee \neg x_3)$, $C_3 = (\neg x_1 \vee \neg x_2 \vee x_3)$ is given in~\cref{fig:grouplivenessNPhardness-example} for illustration. %
    \definecolor{myblue}{RGB}{0,0,255}      %
    \definecolor{mygreen}{RGB}{0,128,0}     %
    \definecolor{mypink}{RGB}{255,0,255}    %
    \definecolor{myorange}{RGB}{255,165,0}  %
    \definecolor{mypurple}{RGB}{128,0,128}  %
    \definecolor{myred}{RGB}{255,0,0}       %
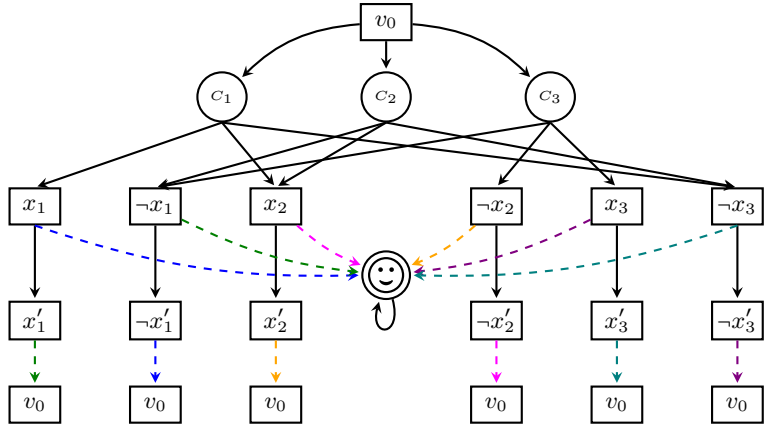
\begin{figure}[t]
    \centering
    \begin{tikzpicture}[player1/.style={draw, thick, rectangle, text width=13pt, text height=6.5pt, align=center}
        ]
        \node[player1] (v0) at (current page.center) {\phantom{$v_0$}};
        
        \node[player0, below=0.4 cm of v0] (C2) {\tiny{$C_2$}};
        \node[player0, left=1.5 cm of C2] (C1) {\tiny{$C_1$}};
        \node[player0, right=1.5 cm of C2] (C3) {\tiny{$C_3$}};
        
        \path[->] (v0) edge[bend right=20] (C1);
        \path[->] (v0) edge (C2); 
        \path[->] (v0) edge[bend left=20] (C3); 

        \node[below=1cm of C2] (middle) {};
        \node[player1, left=1cm of middle] (x2) {\phantom{$x_2$}};
        \node[player1, left=9mm of x2] (nx1) {\phantom{$\neg x_1$}};
        \node[player1, left=9mm of nx1] (x1) {\phantom{$x_1$}};

        \node[player1, right=1cm of middle] (nx2) {\phantom{$\neg x_2$}};
        \node[player1, right=9mm of nx2] (x3) {\phantom{$x_3$}};
        \node[player1, right=9mm of x3] (nx3) {\phantom{$\neg x_3$}};

        \node[xshift=-0.12mm] at (nx3) {$\neg x_3$};
        \node[xshift=-0.12mm] at (nx2) {$\neg x_2$};
        \node[xshift=-0.12mm] at (nx1) {$\neg x_1$};

        \node[] at (v0) {$v_0$};
        \node[] at (x1) {$x_1$};
        \node[] at (x2) {$x_2$};
        \node[] at (x3) {$x_3$};
        
        \path[->] (C1.south) edge (x1.north);
        \path[->] (C1.south) edge (x2.north);
        \path[->] (C1.south) edge (nx3.north);

        \path[->] (C2.south) edge (nx1.north);
        \path[->] (C2.south) edge (x2.north);
        \path[->] (C2.south) edge (nx3.north);

        \path[->] (C3.south) edge (nx1.north);
        \path[->] (C3.south) edge (nx2.north);
        \path[->] (C3.south) edge (x3.north);
 
        \node[player1, below=1cm of x1] (x1p) {\phantom{$x'_1$}};
        \node[player1, below=1cm of x2] (x2p) {\phantom{$x'_2$}};
        \node[player1, below=1cm of x3] (x3p) {\phantom{$x'_3$}};
        \node[player1, below=1cm of nx1] (nx1p) {\phantom{$\neg x'_1$}};
        \node[player1, below=1cm of nx2] (nx2p) {\phantom{$\neg x'_2$}};
        \node[player1, below=1cm of nx3] (nx3p) {\phantom{$\neg x'_3$}};

        \node[player0, below=1.7cm of C2] (smiley) {\phantom{\smiley{}}};
        \node[] at (smiley) {\huge{\smiley{}}};

        \node[xshift=-0.12mm] at (nx1p) {$\neg x'_1$};
        \node[xshift=-0.12mm] at (nx2p) {$\neg x'_2$};
        \node[xshift=-0.12mm] at (nx3p) {$\neg x'_3$};

        \node[] at (x1p) {$x'_1$};
        \node[] at (x2p) {$x'_2$};
        \node[] at (x3p) {$x'_3$};

        \path[->] (x1.south) edge (x1p.north);
        \path[->] (x2.south) edge (x2p.north);
        \path[->] (x3.south) edge (x3p.north);
        \path[->] (nx1.south) edge (nx1p.north);
        \path[->] (nx2.south) edge (nx2p.north);
        \path[->] (nx3.south) edge (nx3p.north);

        \path[thick] (smiley) edge[loop below] (smiley);

        \path[->, thick, dashed, myblue] (x1.south) edge[bend right=10] (smiley.west);     
        \path[->, thick, dashed, mygreen] (nx1) edge[bend right=10] (smiley); 
        \path[->, thick, dashed, mypink] (x2) edge[bend right=10] (smiley); 
        \path[->, thick, dashed, myorange] (nx2) edge[bend left=10] (smiley); 
        \path[->, thick, dashed, mypurple] (x3) edge[bend left=10] (smiley); 
        \path[->, thick, dashed, teal] (nx3.south) edge[bend left=10] (smiley.east);     

        \node[player1, below=0.6cm of x1p] (x1tov0) {\phantom{$v_0$}};
        \node[player1, below=0.6 of nx1p] (nx1tov0) {\phantom{$v_0$}};
        \node[player1, below=0.6 of x2p] (x2tov0) {\phantom{$v_0$}};
        \node[player1, below=0.6 of nx2p] (nx2tov0) {\phantom{$v_0$}};
        \node[player1, below=0.6 of x3p] (x3tov0) {\phantom{$v_0$}};
        \node[player1, below=0.6 of nx3p] (nx3tov0) {\phantom{$v_0$}};

        \node[] at (x1tov0) {$v_0$};
        \node[] at (nx1tov0)  {$v_0$};
        \node[] at (x2tov0)  {$v_0$};
        \node[] at (nx2tov0)  {$v_0$};
        \node[] at (x3tov0)  {$v_0$};
        \node[] at (nx3tov0)  {$v_0$};

        \path[->, thick, dashed, mygreen] (x1p.south) edge (x1tov0.north);
        \path[->, thick, dashed, myblue] (nx1p.south) edge (nx1tov0.north);
        \path[->, thick, dashed, myorange] (x2p.south) edge (x2tov0.north);
        \path[->, thick, dashed, mypink] (nx2p.south) edge (nx2tov0.north);
        \path[->, thick, dashed, teal] (x3p.south) edge (x3tov0.north);
        \path[->, thick, dashed, mypurple] (nx3p.south) edge (nx3tov0.north);

    \end{tikzpicture}\caption{Game $\auggame^\varphi$. %
    Each live group is denoted by dashed edges of a different color, e.g. %
     $H^1_1$ is blue. $\varphi$ has a satisfying assignment $L = \{x_1, x_2, x_3\}$ and thus, the positional strategy with $\stratz(C_1) = x_1, \stratz(C_2)= x_2$ and $\stratz(C_3)= x_3$ is winning.
     }\label{fig:grouplivenessNPhardness-example}
\end{figure}
    
    Game $\auggame^\varphi$ starts from the $\p{1}$ vertex $v_0$. $v_0$ has an outgoing edge to each $\p{0}$ vertex $C_i$, each representing the respective clause in the 3-SAT formula. By taking the edge $(v_0, C_i)$, $\p{1}$ challenges $\p{0}$ to satisfy clause $C_i$. $C_i$ has an outgoing edge to every literal in the clause. By taking edge $(C_i, y)$, $\p{0}$ decides through which literal it will satisfy $C_i$. 
    Then from each $\p{1}$ vertex $y$, there exist two edges: one to $\large{\smiley{}}$, and the other to $y'$. From vertex $y'$ there is only one outgoing edge, i.e. to $v_0$.
    The live groups assert the following condition: Whenever a literal or its negation (say, $y$ or $y'$) is visited infinitely often in the game, $\p{1}$ must take one of the edges $(y, \large{\smiley{}})$ or $(\neg y', v_0)$ infinitely often. This condition forces $\p{1}$ to either for each literal the game visits infinitely often, also visit the negation of the literal infinitely often; or, to go to $\large{\smiley{}}$. %
    
    If $\varphi$ has a satisfying assignment, then the game is won by the positional strategy $\stratz$ that sends each $C_i$ to a literal satisfied in $C_i$ (see~\cref{fig:grouplivenessNPhardness-example}). If $\varphi$ is unsatisfiable, then 
    every positional $\pz$ strategy (it is sufficient to consider positional strategies for $\pz$ as $\auggame^\varphi$ can be viewed as a Rabin game -- see the proof of~\cref{lem:NPinclusion-augmentedparity},  and Rabin games are half positional) there exist $C_i, C_j, y$ with $\stratz(C_i) = y$ and $\stratz(C_j) = \neg y$, or else, $\stratz$ gives a satisfying assignment for $\varphi$. 
    For each such $\stratz$, $\po$ has a winning strategy $\strato$: %
    Whenever $v_0$ is visited it alternates between $\strato(v_0) = C_i$ and $\strato(v_0)= C_j$ and never takes an edge to $\large{\smiley{}}$. A $\stratz \strato$-play neither violates any assumptions nor visits $\large{\smiley{}}$. Thus, $\pz$ has no winning strategy in $\auggame^\varphi$ and therefore, $\po$ wins $\auggame^\varphi$.
    \qed

\smallskip
Since reachability games are a special case of parity games, \cref{lem:NPinclusion-augmentedparity} implies reachability games augmented with live groups lies in $\NP$ and~\cref{lem:NPhardness-augmentedreachability} implies parity games augmented with live groups is $\NP$-hard. 
Hence, combining these results with the ones from \cref{lem:NPinclusion-augmentedparity} and \cref{lem:NPhardness-augmentedreachability} proves \cref{thm:NP-completeness-augmentedreachparity} (see \cref{appendix:thm:NP-completeness-augmentedreachparity} for the complete proof).

\subsection{Singleton-Source Live Groups}\label{section:CNF-edges}

The proof for $\NP$-completeness in \cref{subsec:live-group-def} shows this hardness result already for live groups where each group only contains two edges (indicated by different colors in \cref{fig:grouplivenessNPhardness-example}). Hence, one of the simplest possible generalization of live edges to live groups makes the problem already harder to solve. However, these edges had two different source vertices. The next theorem shows that this distinction in source vertices is actually necessary for $\NP$-hardness. In particular, it shows that \emph{singleton-source live groups}, i.e., live groups $\livegroups$ s.t.\ $\abs{\source(\livegroupSingle)} = 1$ for all $\livegroupSingle\in \livegroups$, can be reduced to games annotated with live edges, and hence can be solved efficiently. The corresponding construction is illustrated with an example in \cref{fig:boolean-live-groups} and proven in \cref{appendix:thm:singleton-live-groups}. 

\begin{restatable}{theorem}{restatesingletonlivegrps}\label{thm:singleton-live-groups}
    Reachability (resp. parity) games augmented with singleton-source live groups can be solved in $\PTIME$ (resp. $\QP$).
\end{restatable}

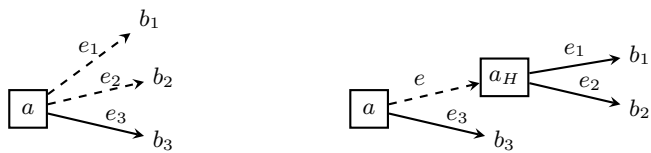
\begin{figure}[t]
    \centering
    \begin{tikzpicture}
        \node[player1] (a) at (0, 0) {$a$};
        \node (b1) at (0.9*\hpos,1*\ypos) {$b_1$};
        \node (b2) at (1*\hpos,0.35*\ypos) {$b_2$};
        \node (b3) at (1*\hpos,-0.35*\ypos) {$b_3$};

        \draw[->,thick,dashed] (a) -- node[above]{$e_1$} (b1);
        \draw[->,thick,dashed] (a) -- node[above,xshift=0.2cm,yshift=-0.03cm]{$e_2$} (b2);
        \draw[->,thick] (a) -- node[yshift=-0.1cm]{$e_3$} (b3);
    \end{tikzpicture}
    \hspace*{2cm}
    \begin{tikzpicture}
        \node[player1] (a) at (0, 0) {$a$};
        \node (b1) at (2*\hpos,0.6*\ypos) {$b_1$};
        \node[player1] (al) at (1*\hpos,0.35*\ypos) {$a_\livegroupSingle$};
        \node (b2) at (2*\hpos,0.*\ypos) {$b_2$};
        \node (b3) at (1*\hpos,-0.35*\ypos) {$b_3$};

        \draw[->,thick] (al) -- node[above]{$e_1$} (b1);
        \draw[->,thick] (al) -- node[above,xshift=0.2cm,yshift=-0.07cm]{$e_2$} (b2);
        \draw[->,thick,dashed] (a) -- node{$e$} (al);
        \draw[->,thick] (a) -- node[yshift=-0.1cm]{$e_3$} (b3);
    \end{tikzpicture}
    \caption{Part of a game augmented with live group $\livegroupSingle=\{e_1,e_2\}$ with $\source(e_1)=\source(e_2)=a$ (left), part of an equivalent game augmented with live edge $e$ (right).}\vspace{-0.5cm}\label{fig:boolean-live-groups}    
\end{figure}

Intuitively, a live group $\livegroupSingle$ with a singleton source models a disjunctive form of transition fairness, while the classical definition of live edge assumptions (as in \cref{def:live-edges}) amounts to a conjunctive form of transition fairness. The construction in \cref{fig:boolean-live-groups} shows that both are equally expressive. It is therefore not surprising that the complexity result from \cref{thm:singleton-live-groups} generalizes to combinations of conjunctions and disjunctions of live edges from the same source in CNF (Conjunctive Normal Form). This is formalized and proven in \cref{appendix:CNFassumptions}.

\subsection{Live Groups in Product Games}\label{sec:product games}

As motivated in the introduction as well as in \cref{example:live-groups}, the use of progress assumptions, especially in the context of CPS control, typically stems from a modelling step where such annotations help to capture progress ensured by the underlying physical system in an abstract manner. In this context, the high-level synthesis game is typically obtained by, first, constructing a game graph $G$ annotated with assumption $\assump$ (as an abstraction of the underlying dynamics), and second, taking the product of the annotated graph $(G,\assump)$ with a parity game $\game'$ obtained from the translation of the  specification $\spec$, which is typically given as an arbitrary LTL formula over the vertex set $V$ of $G$. This process was examplified in \cref{example:live-groups}, illustrating the need for live group assumptions. %

Further, even if the specification $\spec$ is, for example, a simple reachability objective, a product construction might still be needed prior to synthesis, if the controlled system consists of multiple interacting components, some of which annotated with \emph{live edge} assumptions. Then the product of all component models, and hence the final reachability game, is annotated with \emph{live group} assumption. %

While we have seen in the previous subsection that live group assumptions render the respective solution problem  $\NP$-complete (even for reachability games), one could hope that the progress assumptions stemming from the product construction described above, is a more tractable subclass of live group assumptions. 
This would be a reasonable anticipation as the liveness assumptions arising from the product construction can easily be represented as live group assumptions, and these assumptions appear to carry more structure, as they inherit the edges in each live group from a single live edge in one automaton/game during the product construction.
Unfortunately, in what is to come we show that this is not the case. 

Before we move on to \cref{thm:livegroups-product-equivalence} stating the main result of this section, we introduce some concepts needed to state the theorem. In particular, we introduce alternations and labels for games, which are the de-facto standard in synthesis games stemming from CPS design problems, to easily formalize the above discussed product construction in the usual way.

\smallskip
\noindent \textbf{Alternating Games.} 
    A game graph $G = (V, E)$ is called \emph{alternating} if for each $(u,v) \in E$ either $u \in V_0$ and $v \in V_1$ or vice versa. A game $(G, \spec)$ or an augmented game $(G, \spec, \assump)$ is called alternating if $G$ is alternating.%
\footnote{We note that all games can be converted to an equivalent alternating game by at most doubling the size of the vertex and edge sets of the game graph.}

\smallskip
\noindent \textbf{Labeled Games.} 
Let $G = (V, E)$ be a game graph, then the tuple $(G, \labeling: E \to \Sigma)$ is called \emph{the game graph $G$ labeled with $\labeling$}, where the function $\labeling$ assigns an action $a$ to each edge from the finite alphabet $\Sigma$. Similarly, the tuples $(\game = (G, \spec), \labeling)$ and $(\auggame = (G, \spec, \assump), \labeling)$ are called \emph{the games $\game$ and $\auggame$ labeled with $\labeling$}.
For brevity, we denote these tuples by $G_\labeling, \game_\labeling$ and $\auggame_\labeling$, respectively.
In what is to come, we will also look at labeled game graphs augmented with assumptions $(G = (V, E), \assump, \labeling)$, and denote them by $G_{\assump, \labeling}$. 
Intuitively, in a labeled game as given above, $\p{i}$ moves to a vertex $v$ from a vertex $u \in V_i$ with action $a$ (denoted by $u \xrightarrow{a} v$) if and only if there exists an edge $(u,v) \in E$ with $\labeling(u,v) = a$.

\smallskip
\noindent \textbf{Product Games.} 
Let $\game_{\labeling^1} = ((V^1, E^1) , \spec,\labeling^1: E^1 \to \Sigma) $ be an alternating labeled game with initial vertex $v^1_0 \in V^1_1$ and  $G_{\assumplive(E^\ell), \labeling^2}= ((V^2, E^2), \assumplive(E^\ell), \labeling^2 : E^2 \to \Sigma)$ be an alternating labeled augmented game graph with initial vertex $v^2_0 \in V^2_1$. The product $\game_\labeling \times G_{\assumplive(E^\ell), \labeling^2}$ is defined in the usual way to be the alternating augmented game $\product =( ( (V_\product, E_\product), \spec_\product ), \assumpgrlive(\livegroups_\product))$ s.t.\
\begin{compactitem}
    \item $V_\product$ is the cartesian product $V^1_0 \times V^2_0 \cup V^1_1 \times V^2_1$,%
    \item the initial vertex of $\product$ is $(v^1_0, v^2_0)$,
    \item a vertex $(v, v') \in V_\product$ belongs to $\p{0}$ if and only if $(v, v') \in V^1_0 \times V^2_0$, 
    \item an edge $(v, v') \to (w, w')$ is in $E_\product$ iff there exists an action $a \in \Sigma$ for which $v \xrightarrow{a} w \in E^1$ and $v' \xrightarrow{a} w' \in E^2$, %
    \item for each $(v', w') \in E^\ell \subseteq E^2$, there exists a live group $\livegroupSingle_{(v',w')} \in \livegroups_\product$ such that $H_{(v',w')} = \{(v, v')\to (w, w') \in E_\product \mid v, w \in V^1\}$, and
    \item a play $(v_1, v'_1) (v_2, v'_2)  \ldots $ satisfies the specification $\spec_\product $ if and only if the play $v_1 v_2 \ldots$ satisfies $\spec$.
\end{compactitem}

\noindent For a play $\rho = (v_1, v'_1)(v_2, v'_2) \ldots$, we denote the projection of $\rho$ to $V^1$ by $\proj^1(\rho) = v_1 v_2 \ldots$ and the projection to $V^2$ by $\proj^2(\rho) = v'_1 v'_2 \ldots$.

\bigskip
\noindent With this, we are now ready to use the formalized product construction and state the main result of this subsection. 

\begin{restatable}{theorem}{restateproducteq}\label{thm:livegroups-product-equivalence}
 Any alternating augmented game $\auggame = ((V, E), \spec, \assumpgrlive(\livegroups))$ where $\livegroups = \{\livegroupSingle_i\mid i \in [1;m]\}$ with initial vertex $v_\init \in V_1$ can be obtained
 as a product of a labeled non-augmented game $\game_{\labeling^1}$ of same size, and a labeled game graph $G_{\assumplive(E^\ell), \labeling^2}$ of size $m+1$ augmented with live edges $E^\ell$. 
\end{restatable}

\noindent\textit{Proof sketch [Full proof in \cref{appendix:thm:livegroups-product-equivalence}].}
We construct a non-augmented game and the labeled augmented game graph from a given $\auggame = ((V, E), \spec, \assumpgrlive(\livegroups))$ and show that their product is equivalent to the original game. Towards this goal we define $\Sigma:=\{a\}\cup\{h_i|i\in[0;m]\}$ and 
$\game_{\labeling^1} := (((V, E), \spec), \labeling^1: E \to \Sigma)$ s.t.\ $v^1_0 := v_\init$, %
 \begin{equation*}
    \labeling^1:E \to \Sigma \text{ where }\labeling^1(u,v) =\begin{cases} h_i, \text{ if } (u,v) \in H_i \text{ for some } i \in [1;m],\\
     a, \text{ otherwise,} \end{cases}
 \end{equation*}
  where $a$ is a fixed letter different from each $h_i$. Further, define \\ $G_{\assumplive(E^\ell), \labeling^2} := ((V^2, E^2), \assumplive(E^\ell), \labeling^2 :E^2 \to \Sigma)$ s.t.\ $v^2_0 = u_0 \in V^2_1$ and 
 \begin{itemize}
    \item $V^2 = \{u_0\} \cup \{u_*\} \cup \{x_1, \ldots, x_m\}$, where only $u_0$ is a $\p{1}$ vertex,
    \item $E^2 = \{(u_0, x_i), (x_i, u_0) \mid i \in [1;m]\} \cup \{(u_0, u_*), (u_*, u_0)\}$
    \\ where $E^\ell = \{(u_0, x_i) \mid i \in [1;m]\} $,
    \item $\labeling^2(u_0, x_i) = h_i$ and all other edges in $E^2$ are labeled with the letter $a$
 \end{itemize}
(see \cref{fig:productgamegraph} for an illustration). 
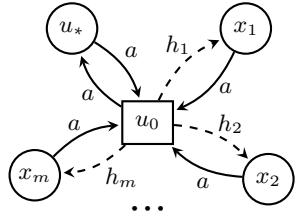
\begin{figure}[t]
    \centering
    \begin{tikzpicture}[player1/.style={draw, thick, rectangle, text width=13pt, text height=9pt, align=center}
        ]
        \node[player1] (v0) {\phantom{$u_0$}};
        
        \node[player0, above left=1cm of v0, xshift=3mm] (v*) {\phantom{$x_1$}};
        \node[player0, above right=1cm of v0] (x1) {\phantom{$x_1$}};
        \node[player0, below right=1cm of v0, yshift=5mm, xshift=3mm] (x2) {\phantom{$x_1$}};
        \node[player0, left=1cm of v0, yshift=-7mm, xshift=2mm] (xm) {\phantom{$x_1$}};
        \node[below=0.7cm of v0] (dots) {\LARGE{$\cdots$}};

          \node[] at (v*) {$u_*$};
          \node[] at (v0)  {$u_0$};
          \node[] at (x1)  {$x_1$};
          \node[] at (x2)  {$x_2$};
          \node[] at (xm)  {$x_m$};

        \path[->, dashed] (v0) edge[bend left=20] node[above, xshift=-0.8mm] {$h_1$} (x1);
        \path[->] (x1) edge[bend left=20] node[right] {$a$} (v0);
        \path[->] (v0) edge[bend left=20] node[below, yshift=0.5mm, xshift=-1mm] {$a$} (v*);
        \path[->] (v*) edge[bend left=20] node[right, yshift=2mm, xshift=-1mm] {$a$} (v0);
        \path[->, dashed] (v0) edge[bend left=20] node[above, xshift=2mm, yshift=-1mm] {$h_2$} (x2);
        \path[->] (x2) edge[bend left=20] node[below] {$a$} (v0);
        \path[->, dashed] (v0) edge[bend left=20] node[right, yshift=-1.6mm] {$h_m$} (xm);
        \path[->] (xm) edge[bend left=20] node[left, xshift=0.8mm, yshift=1.2mm] {$a$} (v0);
    \end{tikzpicture}\vspace{-0.3cm}
    \caption{$G_{\assumplive(E^\ell), \labeling^2}$ : Labeled alternating game graph $(V^2,E^2)$ augmented with live edges (shown by dashed lines) labeled $h_1, \ldots, h_m$.}\vspace{-0.3cm}\label{fig:productgamegraph}
\end{figure}
In order to prove the equivalence of $\auggame$ and $\product$, it is enough to prove the equivalence of plays, i.e. a play
$\rho = (v_0, v'_0) (v_1, v'_1)\ldots$ is winning in $\product$ iff $\proj^1(\rho)$, as a play on the identical game graph $(V,E)$ of $\game_{\labeling^1}$ and $\auggame$, is winning in $\auggame$. 
Let $\rho$ be winning. Then $\proj^1(\rho)$ satisfies $\spec$ or $\rho$ violates a live group assumption 
$H_{(u_0, x_i)}$. This implies that there exists a $w \to v \in H_i$ with $(w, u_0) \to (v, x_i) \in H_{(u_0, x_i)}$ that is enabled infinitely often in $\proj^1(\rho)$, but no edges in $H_i$ are taken infinitely often. Therefore, $\proj^1(\rho)$ violates the group liveness assumptions in $\auggame$, and thus, it is winning in $\auggame$.

Now let $\rho$ be winning for $\po$. Then $\proj^1(\rho)$ violates $\spec$ and $\rho$ satisfies all live group assumptions $H_{(u_0, x_i)}$. This translates to $\proj^1(\rho)$ similar to before: for each $i \in [1;m]$, if an edge $w \to v \in H_i$ is visited infinitely often, then an edge in $H_i$ is taken infinitely often. Thus, $\proj^1(\rho)$ is winning for $\po$ in $\auggame$.\qed

\smallskip
With \cref{thm:NP-completeness-augmentedreachparity} and \cref{thm:livegroups-product-equivalence} in place, it follows that games resulting from the product of multiple graphs -- some of which annotated with live edges -- renders the resulting augmented synthesis problem $\NP$-complete.

\subsection{A Remark on Co-Live Groups}
As we considered live groups in the last sections, it would be natural to also consider co-live groups, i.e., groups of edges where the assumption is to ensure that at least one of them is co-live.
However, unlike in the case of live edge assumptions, taking the product of a game with co-live edges and another non-augmented game graph does actually result in a game augmented with co-live \emph{edges} not co-live \emph{groups}.
This is due to the fact that \emph{every} edge in the product game which originates from a co-live edge in the original game, is only allowed to be taken finitely often.
Furthermore, given a live group $\livegroupSingle$, its dual, $\neg\LTLalways\LTLeventually\livegroupSingle$ can actually be expressed by a set of co-live edges, since
\[\neg\LTLalways\LTLeventually\livegroupSingle
= \neg\LTLalways\LTLeventually \bigvee_{e\in \livegroupSingle} e
= \neg \bigvee_{e\in \livegroupSingle} \LTLalways\LTLeventually e
= \bigwedge_{e\in \livegroupSingle} \neg\LTLalways\LTLeventually e.\]
Hence, co-live groups do not seem relevant in the considered context.

%% file: sections/progress-groups.tex
\section{Persistent Live Group Assumptions}\label{section:progress-group}
Given the practical relevance of product constructions in CPS synthesis problems, as outlined in \cref{sec:product games}, the negative result of \cref{thm:NP-completeness-augmentedreachparity} and \cref{thm:livegroups-product-equivalence} is rather discouraging, as the more tractable subclass in \cref{section:CNF-edges} seems rather restrictive.

In contrast, this section discusses a different version of live groups, called \emph{persistent live groups}, which originally appeared in the work of Ozay et al.~\cite{progress_groups,NilssonOL17} in the context of abstraction-based control design (ABCD) and which (i) have very nice computational properties while (ii) being closed under product constructions. Before formally defining this assumption class and formalizing its properties, we give an example\footnote{See \cite{context-triggeredABCD} for a more in-depth version of this example.} from ABCD to motivate their relevance.

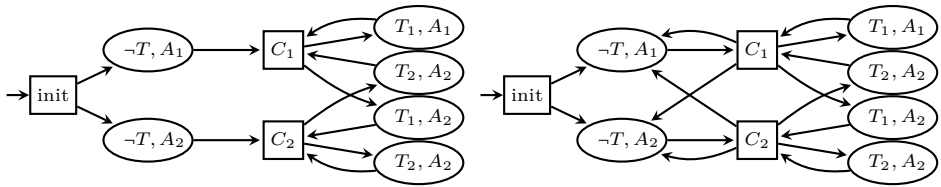
\begin{figure}[t]
    \scriptsize
    \centering
    \begin{tikzpicture}
        \node[player1] (i) at (0.3*\hpos, 0) {init};
        \node[bplayer0] (xa1) at (\hpos, 0.5*\ypos) {$\neg T,A_1$};
        \node[bplayer0] (xa2) at (\hpos, -0.5*\ypos) {$\neg T, A_2$};
        \node[player1] (c1) at (2*\hpos, 0.5*\ypos) {$C_1$};
        \node[player1] (c2) at (2*\hpos, -0.5*\ypos) {$C_2$};
        \node[bplayer0] (t1a1) at (3*\hpos, 0.75*\ypos) {$T_1,A_1$};
        \node[bplayer0] (t2a1) at (3*\hpos, 0.25*\ypos) {$T_2,A_2$};
        \node[bplayer0] (t1a2) at (3*\hpos, -0.25*\ypos) {$T_1,A_2$};
        \node[bplayer0] (t2a2) at (3*\hpos, -0.75*\ypos) {$T_2,A_2$};

        \draw[<-,thick] (i.west) --+ (-0.3,0);
        \path[->] (i) edge (xa1) edge (xa2);
        \path[->] (xa1) edge (c1);
        \path[->] (xa2) edge (c2);
        \path[->] (c1) edge (t1a1) edge[bend left=-10] (t1a2);
        \path[->] (c2) edge[bend left=10] (t2a1) edge (t2a2);
        \path[->] (t1a1) edge[bend left=-20] (c1);
        \path[->] (t2a1) edge (c1);
        \path[->] (t1a2) edge (c2);
        \path[->] (t2a2) edge[bend left=20] (c2);

        \node[player1] (i) at (3.8*\hpos, 0) {init};
        \node[bplayer0] (xa1) at (4.5*\hpos, 0.5*\ypos) {$\neg T,A_1$};
        \node[bplayer0] (xa2) at (4.5*\hpos, -0.5*\ypos) {$\neg T, A_2$};
        \node[player1] (c1) at (5.5*\hpos, 0.5*\ypos) {$C_1$};
        \node[player1] (c2) at (5.5*\hpos, -0.5*\ypos) {$C_2$};
        \node[bplayer0] (t1a1) at (6.5*\hpos, 0.75*\ypos) {$T_1,A_1$};
        \node[bplayer0] (t2a1) at (6.5*\hpos, 0.25*\ypos) {$T_2,A_2$};
        \node[bplayer0] (t1a2) at (6.5*\hpos, -0.25*\ypos) {$T_1,A_2$};
        \node[bplayer0] (t2a2) at (6.5*\hpos, -0.75*\ypos) {$T_2,A_2$};

        \draw[<-,thick] (i.west) --+ (-0.3,0);
        \path[->] (i) edge (xa1) edge (xa2);
        \path[->] (xa1) edge (c1);
        \path[->] (xa2) edge (c2);
        \path[->] (c1) edge (t1a1) edge[bend left=-10] (t1a2);
        \path[->] (c2) edge[bend left=10] (t2a1) edge (t2a2);
        \path[->] (t1a1) edge[bend left=-20] (c1);
        \path[->] (t2a1) edge (c1);
        \path[->] (t1a2) edge (c2);
        \path[->] (t2a2) edge[bend left=20] (c2);

        \path[->] (c1) edge[bend left=-20] (xa1) edge (xa2);
        \path[->] (c2) edge[bend left=20] (xa2) edge (xa1);
    \end{tikzpicture}
    \caption{Two abstract game graphs to model the robot motion control problem from \cref{example:progress-groups}: without persistent live groups (left) and with persistent live groups $(\perssource_i,\persedges_i,\perstarget_i) = \left(V,E\cap (V\times \{C_i\}),\{(T_i,\cdot)\}\right)$ for each $i\in[1;2]$ (right).}\vspace{-0.5cm}\label{fig:progress-groups}
\end{figure}
\begin{example}\label{example:progress-groups}
    Consider a robot that should be controlled to react to changes in its desired goal location. For simplicity, suppose there are only two (non-intersecting) target locations $T_1$ and $T_2$ and there exist (already designed) low-level controllers $C_1$ and $C_2$ which ensure that, whenever a target is persistently activated, the robot will reach the respective target region. 
    The high-level synthesis problem is to design a strategy to trigger the low level controllers $C_i$ in response to the currently activated target (signaled by proposition $A_i$ set to 'true' by the environment), such that the target region $T_i$ is actually reached, if the respective target is persistently activated. 
While the correct strategy is obvious in this case -- choose $C_i$ iff $A_i$ is true -- constructing a correct abstract game graph that returns this strategy and allows to conclude that $T_i$ is actually reached under this strategy, is surprisingly cumbersome. 

Two possible abstract game graphs are depicted in \cref{fig:progress-groups}. Here, $\po$ chooses the activated target (i.e., which (single) $A_i$ is 'true') and in which region of the state space the robot currently is (either in one of the target regions $T_i$ or in the non-target region $\neg T$). Further, $\pz$ chooses which controller $C_i$ to activate. In favor of readability, we only depict the $\pz$ edges corresponding to the strategy of choosing $C_i$ iff $A_i$ is activated. 

In \cref{fig:progress-groups} (left) the controller only allows the environment to activate a new target after it reached a target location, while \cref{fig:progress-groups} (right) models the more realistic scenario that the environment might activate a different target at any point, which then allows $\pz$ to activate the other controller, even if the robot is still in $\neg T$.
    However, with the additional edges in \cref{fig:progress-groups} (right), $\po$ can force the play to loop between state $(\neg T, A_i)$ and $C_i$, thereby preventing the play from ever reaching $T_i$ even if $A_i$ is persistently active. In the physical system, however, we know that the robot will reach $T_i$ if $C_i$ is persistently used.     
    Intuitively, this assumption can be expressed by a \emph{persistent live group} which models that in the source region $\perssource_i = V$ (complete vertex set in this case), if $\pz$ persistently chooses edges from $\persedges_i = E\cap (V\times \{C_i\})$ (representing the choice of controller $C_i$), then the play will eventually reach the target region $\perstarget_i = \{(T_i,\cdot)\}$. This augmentation allows to synthesize the correct strategy from the abstract augmented game.
\end{example}

\begin{definition}\label{def:perslivegroups}
Given a game graph $\gamegraph = (V,E)$, a \emph{persistent live group} is a tuple $(\perssource,\persedges,\perstarget)$ consisting of sets $\perssource,\perstarget\subseteq V$ and $\persedges\subseteq E_0$ such that $\perstarget\subseteq \perssource$.
The assumption represented by such a persistent live group is expressed by the LTL formula
\begin{equation}\label{eq:assumpPers}
    \assumpPers(\perssource,\persedges, \perstarget) \coloneqq \LTLalways \big(\LTLalways (\perssource\wedge \assumpC(\persedges)) \Rightarrow \LTLeventually \perstarget\big),  
\end{equation}
where $\assumpC(\persedges) \coloneqq \bigwedge_{(u,v)\in\persedges} u \Rightarrow \LTLnext v$.
Furthermore, the assumptions represented by a set $\perslivegroups$ of persistent live groups is denoted by $\assumpPers(\perslivegroups) \coloneqq \bigwedge_{(\perssource,\persedges,\perstarget)\in\perslivegroups} \assumpPers(\perssource,\persedges, \perstarget)$.
Moreover, we write games augmented with set $\perslivegroups$ of persistent live groups to refer to the augmented games $\auggame = (\gamegraph,\spec,\assumpPers(\perslivegroups))$.
\end{definition}
Intuitively, $\assumpC(\persedges)$ ensures that edges in $\persedges$ are chosen when possible, as this is only possible for $\p{0}$ vertices in $\perssource$.
Furthermore, \eqref{eq:assumpPers} ensures that 
if $\p{0}$ satisfies the safety constraints as in left side of the implication, i.e., persistently choosing the edges in $\persedges$ from the source vertices $\perssource$, will eventually make progress and reach a vertex in $\perstarget$.
Moreover, once it visits $\perstarget$, $\p{0}$ can choose to not satisfy the safety constraint anymore unless she wants to visit $\perstarget$ again
(see \cite{context-triggeredABCD} for more details).

In particular, such persistent live groups are helpful to $\pz$ only if she satisfies the safety constraints described by it.
Intuitively, this safety part of the live group makes games augmented with such assumptions easier to solve -- $\p{0}$ needs to stick to a particular live group by persistently choosing the corresponding edges until it reaches the corresponding targets. In (normal) live groups, $\p{0}$ can visit the sources of multiple live groups infinitely often, and hence potentially activated multiple ones at the same time.
Furthermore, unlike live groups, $\p{0}$ can utilize persistent live groups in series -- one after another -- as they are safety-type constraint that are not conditioned on visiting vertices infinitely often.
Using this observation, an algorithm to solve a restricted version of augmented games with persistent live groups was already provided in \cite{NilssonOL17}. 
In a recent work, Nayak et al.~\cite{context-triggeredABCD} provide another algorithm for the general case. 
Their algorithm works in polynomial time for reachability objectives but is exponential for parity objectives.
In this section, we will show that using their algorithm for reachability objectives, one can also obtain a quasi-polynomial algorithm for parity objectives.

Before going into the details of the algorithm, let us show a quick remark about persistent live groups in product games. By considering the same scenario as in \cref{sec:product games} but starting with a game graph $\gamegraph$ annotated with a \emph{persistent live group} (instead of live edges as in \cref{sec:product games}), the product of this annotated graph with any other graphs or games (also only potentially augmented with persistent live groups), still results in a product game augmented by \emph{persistent live groups}. This is due to the fact, that the definition of this assumption class already accounts for groups of vertices, which is retained under product constructions.
\begin{remark}
    Persistent live groups assumptions are closed under product (in the sense of  \cref{sec:product games}).
\end{remark}

\subsection{Augmented Reachability Games}
We first consider the augmented reachability game $\auggame=(\gamegraph, \LTLeventually T,\assumpPers(\perslivegroups))$ with persistent live group assumptions. Following \cite{context-triggeredABCD}, the recursive algorithm that solves such an augmented reachability game $\auggame$ is given in \cref{alg:reachPers}. The main idea of the algorithm is to first compute the set of vertices $A$ from which $\p{0}$ can reach $T$ even without the help of any persistent live group assumptions (\cref{alg:reachPers:Attro}) along with the corresponding strategy $\sigma$ for $\p{0}$ (\cref{alg:reachPers:stratA}). Afterwards, the algorithm computes the set of states $B$ from which $\p{0}$ has a strategy (i.e.\ $\sigma_B$) to reach $A$ with the help of a persistent live group (lines~\ref{alg:reachPers:if}-\ref{alg:reachPers:computeB}). If this set $B$ enlarges the winning state set $A$ (\cref{alg:reachPers:B}), we use recursion to solve another such augmented reachability game with target $T:=A\cup B$ (\cref{alg:reachPers:recursion}).

\begin{algorithm}[t]
    \caption{$ \reachPers(\gamegraph, T, \perslivegroups) $}\label{alg:reachPers}
    \begin{algorithmic}[1]
        \Require An augmented reachability game $\auggame = (\gamegraph,\LTLeventually T,\assumpPers(\perslivegroups))$
        \Ensure Winning region and memoryless winning strategy in the augmented game $\auggame$  
        \State Initialize a random $\p{0}$ strategy $\strat$
        \State $A,\strat_A \gets \Attrz{\gamegraph}{T}$\label{alg:reachPers:Attro}
        \State $\strat(v) \gets \strat_A(v)$ for every $v\in A\setminus T$\label{alg:reachPers:stratA}
        \For{$(\perssource,\persedges,\perstarget) \in \perslivegroups$} 
            \If{$(\perssource\setminus A) \cap \pre(A) \neq \emptyset $} \label{alg:reachPers:if}
                \State $B,\strat_B \gets \solve(\gamegraph|_{\persedges}, \spec_B)$ with $\spec_B = \LTLeventually A \vee \LTLalways (\perssource\setminus\perstarget)$\label{alg:reachPers:computeB}
                \If{$B \not\subseteq A$}\label{alg:reachPers:B}
                    \State $\strat(v) \gets \strat_B(v)$ for every $v\in B\setminus A$\label{alg:reachPers:stratB}
                    \State $C,\strat_C\gets \reachPers(\gamegraph, A\cup B, \perslivegroups)$ 
                    \State $\strat(v) \gets \strat_C(v)$ for every $v\in C\setminus (A\cup B)$
                    \State \Return $(C,\strat)$\label{alg:reachPers:recursion}
                \EndIf
            \EndIf
        \EndFor
        \State \Return $A,\strat$ \label{alg:reachPers:end}
    \end{algorithmic}
\end{algorithm}

Within \cref{alg:reachPers}, we use the following notation. Given a game graph $\gamegraph = ( V, E)$ and a persistent live group $(\perssource,\persedges,\perstarget)$, we write $\gamegraph|_{\persedges}$ to denote the restricted game graph $( V, E')$ such that $ E'\subseteq E$ and for every edge $e = ( u, v)\in  E'$, either $e\in\persedges$ or there is no edge in $\persedges$ starting from $ u$.
Furthermore, we use the function $\solve(\gamegraph,\spec)$ to solve the game $(\gamegraph,\spec)$ with $\spec = \LTLeventually A \vee \LTLalways S$ for some $A,S\subseteq V$, which can be done by reducing it to safety game (see \cref{rem:safety} in \cref{appendix:progress-group}).
With this, we can state the result for augmented reachability game given in~\cite{context-triggeredABCD}.
 
\begin{restatable}{proposition}{restateprogressgroup}{\cite[Theorem 2]{context-triggeredABCD}}\label{thm:persGame}
Given an augmented game  $\auggame = (\gamegraph,\spec,\assumpPers(\perslivegroups))$ with game graph $\gamegraph= (V,E)$, specification $\spec = \LTLeventually T$, and persistent live groups $\perslivegroups$, the algorithm $\reachPers(\gamegraph, T, \perslivegroups)$ returns the winning region and a memoryless winning strategy in game $\auggame$.
The algorithm terminates in $\mathcal{O}(\abs{\perslivegroups}\cdot\abs{ V}\cdot\abs{ E})$ time.
\end{restatable}

We re-state the proof of \cref{thm:persGame} in our notation in \cref{appendix:progress-group}.
Using this result, \cref{problem:main} for such augmented reachability games can be answered as follows.
\begin{corollary}\label{corollary:progress-reach-ptime}
    Reachability games augmented with persistent live group assumptions can be solved in $\PTIME$. 
\end{corollary}

\subsection{Augmented Parity Games}\label{sec:progress-groups-parity}
To solve parity games augmented with persistent live groups, we use the fact that most of the algorithms to solve (non-augmented) parity games are based on the attractor functions $\Attrz{\gamegraph}{T}$ and $\attrz{\gamegraph}{T}$ (as defined in \cref{section:prelim}). 
Furthermore, the only difference between the attractor function $\Attrz{\gamegraph}{T}$ and our new function $\reachPers(\gamegraph,T,\perslivegroups)$ from \cref{alg:reachPers} is the utilization of augmented persistent live groups to solve reachability games.
Hence, in many of these algorithms for parity games, one can simply replace every use of $\Attrz{\gamegraph}{T}$ with $\reachPers(\gamegraph,T,\perslivegroups)$ to obtain an algorithm to solve parity games augmented with persistent live groups.
In particular, the authors in \cite{context-triggeredABCD}, obtained an exponential algorithm by using $\reachPers(\gamegraph,T,\perslivegroups)$ in Zielonka's algorithm~\cite{Zielonka98}.
In this section, we show that applying the same technique to the quasi-polynimial algorithm given by Lehtinen et al.~\cite{LehtinenPSW22} and Parys~\cite{quasiZielonka} gives a quasi-polynomial algorithm for such augmented games.

As the algorithm given by Lehtinen et al.~\cite{LehtinenPSW22,quasiZielonka} recursively solves the parity games restricted to a smaller set of vertices, we need to define how the persistent live groups are transformed when we restrict the game graph in such a way.
In addition, we also need to ensure that the transformed persistent live groups still express the same assumption w.r.t.\ the restricted game graph.
Intuitively, one easy way to do this is by restricting all three sets $\perssource$, $\persedges$, $\perstarget$ of a persistent live group to the set of vertices and edges in the restricted game graph.
However, if there is an edge $e=(v,w)$ in $\persedges$, then the constraint $\assumpC(\persedges)$ enforces that edge $e$ is taken from vertex~$v$. If the restricted game graph contains $v$ but not the edge $e$, then to satisfy the constraint $\assumpC(\persedges)$, we need to ensure that vertex $v$ is not visited.
Hence, we need to removed such vertices from $\perssource$ in the transformed persistent live groups. This is formalized below.
\begin{definition}
    Given a game graph $\gamegraph = (V,E)$ augmented with a set $\perslivegroups$ of persistent live groups, and a set $U\subseteq V$, we define the set of persistent live groups restricted to $U$ as $\perslivegroups|_U = \{(\perssource|_U,\persedges|_U,\perstarget|_U) \mid (\perssource,\persedges,\perstarget)\in \perslivegroups\}$ s.t.\
    \begin{align*}
        \perstarget|_U = \perstarget \cap U,\quad
        \persedges|_U = \{(u,v)\in E \mid u,v\in U\},\quad
        \perssource|_U = (\perssource\cap U) \setminus (\source(\persedges)\setminus\source(\persedges')).
    \end{align*}
    Furthermore, the restriction applies as usual to game graph and parity objectives, i.e., $\gamegraph|_U = (U,E\cap U\times U)$ and parity objective $\parity(\priority|_U)$ is defined such that $\priority|_U(u) = \priority(u)$ for all $u\in U$.
\end{definition}
One can show that $\perslivegroups|_U$ indeed captures the same assumption as $\perslivegroups$ w.r.t.\ the game restricted to $U$ as formalized below and proven in \cref{appendix:lemma-prog-restricted}.
\begin{restatable}{lemma}{restateprogressgrouplemma}\label{lemma:prog-restricted}
    Given an augmented parity game $\auggame = (\gamegraph,\spec,\assumpPers(\perslivegroups))$ and a set $U\subseteq V$, let $\play$ be a play in $\gamegraph|_U$.
    Then, $\play$ is winning in augmented game $\auggame|_U =(\gamegraph|_U,\spec|_U,\assumpPers(\perslivegroups|_U))$ if and only if it is winning in augmented game $\auggame$. 
\end{restatable}

With this well-defined restrictions, one can replace every use of $\Attrz{\gamegraph}{T}$ with $\reachPers(\gamegraph,T,\perslivegroups)$ in the algorithm given by Lehtinen et al.~\cite{LehtinenPSW22,quasiZielonka} to obtain a quasi-polynomial result for parity games augmented with persistent live groups.
Although the proof of correctness is almost identical to the proof given by Parys~\cite{quasiZielonka}, we have provided the adapted algorithm and its proof of correctness in \cref{appendix:quasi-progress-groups}.

\begin{theorem}\label{theorem:pers-solve_augmented_parity}
    Parity games augmented with persistent live group assumptions can be solved in $\QP$. 
\end{theorem}

%% file: sections/appendix.tex
\section{Proof of \cref{thm:colive}}\label{appendix:co-live}
\restatecolive*
\begin{proof}
    Let us first show that vertices in $A$ are winning for $\p{1}$. As $A = \attro{\gamegraph}{\wino'}$, $\p{1}$ has a (memoryless) strategy $\strato$ to reach the set $\wino'$ from every vertex in $A$. Furthermore, since $\wino'$ is the winning region for $\p{1}$ in game $(\gamegraph',\spec)$, $\p{1}$ has a (memoryless) strategy $\strato'$ that does not use any co-live edge and is winning from every vertex in $\wino'$. Hence, from every vertex in $A$, $\p{1}$ can use strategy $\strato$ to reach $\wino'$ and then use strategy $\strato'$ to win the game without using the co-live edges anymore. Therefore, $\p{1}$ has a strategy to win the game from $A$ without using the co-live edges in $\colivegroup$ infinitely often. 

    Now, we prove the theorem  using induction on $\abs{V}$. 
    Base case is trivial when $\abs{V} = 0$. 
    For the induction case, assume that the algorithm indeed returns the winning region and a memoryless strategy for games with $\abs{V} < n$.
    Now, suppose $\abs{V}=n$, then, we have two cases, either algorithm terminates in \cref{alg:colive:return1} or in \cref{alg:colive:return2}.

    \paragraph*{Case 1:} Suppose the algorithm terminated and returned $(B,\stratz')$ in \cref{alg:colive:return1}. Then either $B=\emptyset$ or $B=V$. 

    If $B=\emptyset$, then $A = V$. Hence, $\p{1}$ has a winning strategy from every vertex that uses the co-live edges finitely often. Therefore, $\p{0}$'s winning region is empty.

    Now, if $B=V$, then $\winz' = V$. Hence, strategy $\stratz'$ is winning for $\p{0}$ from every vertex in game $(\gamegraph',\spec)$. It is enough to show that $\stratz'$ is also winning for $\p{0}$ from every vertex in game $(\gamegraph,\spec,\assumpcolive(\colivegroup))$. Let $\play$ be a $\stratz'$-play. If $\play$ uses the co-live edges in $\colivegroup$ infinitely often, then by definition, it is winning. If $\play$ uses the co-live edges finitely many times, then there exists a suffix $\play'$ of $\play$ that is compliant with $\stratz'$ and never use any edge in $\colivegroup$. So, $\play'$ is winning in $(\gamegraph',\spec)$ and hence, satisfies the parity objective $\spec$. As parity objectives are tail-properties, $\play$ also satisfies $\spec$ and hence, it is winning.

    \paragraph*{Case 2:} Suppose the algorithm terminated in \cref{alg:colive:return2} and returned $(\winz,\stratz)$. Since $B \neq V$, we have $\abs{B} < \abs{V}$. So, by the induction assumption, $(\winz,\stratz)$ are the winning region and memoryless winning strategy in game $\auggame|_B$.
    Since, every vertex in $A$ is winning for $\p{1}$, the winning region of $\p{0}$ is a subset of $B$. Furthermore, as every vertex in $B\setminus \winz$ is winning for $\p{1}$ in game $\auggame|_B$, it is also winning for $\p{1}$ in $\auggame$. Hence, it is enough to show that $\stratz$ is winning for $\p{0}$ from every vertex in $\winz$. Let $\play$ be a $\stratz$-play starting from some vertex in $\winz$. As $A = \attro{\gamegraph}{\wino'}$ and $\winz\cap A = \emptyset$, $\p{1}$ cannot force a play from $\winz$ to visit $A$. Hence, $\play$ stays inside $B$. As $\stratz$ is a winning strategy for $\auggame|_B$, $\play$ is a winning play. Therefore, $\winz$ is the winning region for $\p{0}$ and $\stratz$ is a (memoryless) winning strategy for her.\qed
\end{proof}

\section{Proof of $\NP$-Completeness Results for Live Groups}\label{appendix:lem:NPinclusion-augmentedparity}\label{appendix:thm:NP-completeness-augmentedreachparity}
Let us first prove the lemmas (\cref{lem:NPinclusion-augmentedparity} and \cref{lem:NPhardness-augmentedreachability}) leading to the $\NP$-completeness result.
\restatelivegrouplemma*
\begin{proof}
    We will give a polynomial time reduction from a parity game augmented with live groups to a Rabin game. As Rabin games are known to be $\NP$-complete~\cite{EmersonJ91,EmersonJ99,Thomas97}, this will imply the desired inclusion result.
    Let $\auggame =(\gamegraph = (V,E),\spec = \parity(\priority), \assumpgrlive(\livegroups) )$ be an augmented parity game with priority function $\priority\colon V\rightarrow [0;d]$.
    As Rabin pairs are defined on vertices and not edges, in order to represent live groups as Rabin pairs, while reducing $\auggame$ to a Rabin game $\game'$, we need to add an additional vertex $e$ for every edge $(u,v)$ that lie in a live group, and edges $ u \to e \to v$ to $\game'$.
    Then, the game graph $G' = (V', E')$ where $V' = V \uplus E$ and $E' = \{ (u,e), (e,v) \mid (u,v) \in E\}$ is large enough to reduce any $\auggame$ into, i.e. even one in which every edge lies in a live group.
    Then we define the Rabin game $\game' = (G', \spec' = \rabin(\Omega_1 \cup \Omega_2))$
    where the Rabin pairs in $\Omega_1 = \{(\source(\livegroupSingle),\livegroupSingle) \mid \livegroupSingle\in\livegroups\}$ represent the live group conditions and
    the Rabin pairs in $\Omega_2 = \{(P_{2i},\cup_{j>2i} P_j)\mid 0\leq 2i\leq d\}$ where $ P_i=\{q\in Q\mid \priority(q)=i \}$ represent the parity conditions.
    For a play $\rho = v_1 \to v_2 \to v_3 \to \ldots $ in $G$, call the play $\rho' = v_1 \to (v_1, v_2) \to v_2 \to (v_2, v_3) \to v_3 \to \ldots$ in $G'$ the projection of $\rho$ to $G'$ and similarly, call $\rho$ the projection of $\rho'$ to $G$. Due to our construction, the projections are well defined on both directions. 
    We will show that $\rho$ is winning in $\auggame$ if and only if $\rho'$ is winning in $\game'$.
    
    Let $\rho'$ be winning in $\game'$. Then there exists a Rabin pair $(F_i, R_i)$ that is satisfied by $\rho'$. If the pair is in $\Omega_1$, then there exists a live group $\livegroupSingle \in \livegroups$ such that $\rho'$ visits a source vertex $u$ for $e = (u,v) \in \livegroupSingle$ infinitely often and visits vertex $e$ only finitely often.
    This implies that in $\rho$, $u$ is visited infinitely often and the edge $e$ is taken only finitely often. Hence, $\rho'$ violates the live group assumption $\livegroupSingle$.
    Otherwise, if the pair $(F_i, R_i) \in \Omega_2$, then for some $i\in [0; d/2]$, $\play$ visits $P_{2i}$ infinitely often and visits $\cup_{j>2i} P_j$ only finitely often. This implies that $\rho$ visits the vertices with priority $2i$ infinitely often and the ones with higher priorities only finitely often. 
    In both cases, $\rho$ is winning in $\auggame$.

    Now let $\rho$ be winning in $\auggame$. We can similarly see that if $\rho$ is won due to the violation of a live group assumption $\livegroupSingle$, then $\rho'$ satisfies a Rabin pair in $\Omega_1$. Otherwise if $\rho$ is won due to the the parity condition, then $\rho'$ satisfies a Rabin pair in $\Omega_2$. Either way, $\rho'$ is winning in $\game'$.\qed
\end{proof}

\restatelivegrouplemmab*

\begin{proof}
 Given the construction of $\auggame^\varphi$ in \cref{subsec:live-group-def} it remains to show that $\varphi$ has a satisfying assignment iff $v_0$ is won in $\auggame^\varphi$.
    \paragraph*{$(\Rightarrow)$}First suppose $\varphi$ has a satisfying assignment. Let $L$ be a set of literals the assignment sets $\true$ to. Clearly, for each $i \in [1; m]$, either $x_i$ or $\neg x_i$ is in $L$. Furthermore, for each clause $C_i$, there exists a literal in $C_i$ that is in $L$.
    Now set the following positional strategy to $\p{0}$: For each vertex $C_i$, take an edge $(C_i, y)$ where $y \in L$. We claim this strategy is winning from $v_0$.
    Observe that in any play $\rho$ compliant with this strategy for each literal $y$, one of $y, \neg y$ is never visited.
    Moreover, there exists at least one literal $y$ visited infinitely often. 
    If $y = x_j$, then the edge $(\neg x'_j, v_0)$ cannot be taken, and thus $\p{1}$ has to take the edge $(x_j, \large{\smiley{}})$ to make $H_j^1$ hold. Similarly, if $y = \neg x_j$, he has to take $(\neg x_j, \large{\smiley{}})$ for $H^2_j$ to hold.
    Either way, $\p{0}$ wins the game.

    \paragraph*{$(\Leftarrow)$}For the other direction, recall from the proof of \cref{lem:NPinclusion-augmentedparity} that every parity (and thus, reachability) game augmented with live groups can be viewed as a Rabin game. In the proof of the converse direction, we will use the fact that Rabin games are half-positional~\cite{EJ99}. That is, if $\p{0}$ has a winning strategy from a vertex in a Rabin game, then it also has a positional one.
    Suppose $\varphi$ has no satisfying assignment. We will show that for each positional strategy of $\p{0}$, there exists a $\p{1}$ strategy such that a play that starts from $v_0$ and is compliant with both strategies is winning for $\p{1}$. 
    Let $\strat: V_0 \to V$ be a positional $\p{0}$ strategy. Then, since the successors of each clause $C_i$ are its literals, $\strat$ cannot send each clause to a literal in a non-conflicting manner, or else $\strat$ would give a satisfying assignment for $\varphi$. That is, there exists $i_1, i_2 \in [1;m]$ such that $\strat(C_{i_1}) = x_j$ and $\strat(C_{i_2}) = \neg x_j$ for some literal $x_j$. %
    Consider the $\p{1}$ strategy that alternates between $(v_0, C_{i_1})$ and $(v_0, C_{i_2})$ whenever $v_0$ is visited and takes the edges $(x_j, x'_j)$ and $(\neg x_j, \neg x'_j)$. A play compliant with both strategies will never reach $\large{\smiley{}}$, only see literals $x_j$ and $\neg x_j$ infinitely often, and thus will only have to satisfy $H^1_j$ and $H^2_j$. Both of these 
    live groups are satisfied due to the edges $(\neg x'_j, v_0)$ and $(x'_j, v_0)$ being taken infinitely often. That is, $\p{0}$ has no winning strategy from $v_0$, and thus $v_0$ is won by $\p{1}$.\qed
\end{proof}

Now, let us prove the $\NP$-completeness result using the lemmas.
\restatenpcompletenesslivegrps*
\begin{proof}
    Observe that reachability games are a special case of parity games: In the reachability game we assign priority $2$ to each vertex in the reachability set and turn them into sink vertices by removing their outgoing edges and adding self-loops. We assign priority $1$ to every other vertex. The resulting game is an equivalent parity game.
    It is easy to see that the same conversion give that reachability games augmented with live groups are a special case of parity games augmented with live groups.
    Consequently, \cref{lem:NPinclusion-augmentedparity} implies that reachability games augmented with live groups lies in $\NP$ and \cref{lem:NPhardness-augmentedreachability} implies that parity games augmented with live groups is $\NP$-hard. Combining these results with the existing ones from \cref{lem:NPinclusion-augmentedparity} and \cref{lem:NPhardness-augmentedreachability} gives us the $\NP$-completeness result for both games.\qed
 \end{proof}

\section{Singleton-Source Live (CNF) Groups Assumptions}\label{appendix:CNFassumptions}\label{appendix:thm:singleton-live-groups}
Let us first prove \cref{thm:singleton-live-groups} showing games augmented with singleton-source live groups can be solved in same time complexity class as their non-augmented versions.
\restatesingletonlivegrps*
\begin{proof}
The proof proceeds by translating a game augmented with singleton-source live groups into a different game augmented with live eges. For live group $\livegroupSingle \in \livegroups$, let $a = \source(\livegroupSingle)$. Then  add one new vertex $a_\livegroupSingle$ (with priority~$0$ in case of a parity specification) to the game as well as \emph{normal edges} $(a_\livegroupSingle, v_i)$ for each $(a, v_i) \in \livegroupSingle$ and replace all the edges $e_i = (a, v_i) \in \livegroupSingle$ with a single live edge $e = (a, a_\livegroupSingle)$. This ensures that in the new game whenever $a$ is taken infinitely often, the unique live outgoing edge $e$ of $a$ will be taken infinitely often. From $a_\livegroupSingle$, $\po$ is free to choose any edge from $\livegroupSingle$ to take. As taking any one of them enables it to satisfy the live group assumption $\livegroupSingle$, both games are equivalent. This construction is illustrated on an example in \cref{fig:boolean-live-groups}.\qed
\end{proof}

As stated before, the idea can be generalized to a class of fairness assumptions expressed by combinations of edges from same source in CNF forms as formalized below.
\input{sections/CNF-edges.tex}

\section{Proof of \cref{thm:livegroups-product-equivalence}}\label{appendix:thm:livegroups-product-equivalence}
\restateproducteq*
\begin{proof}
    We will use the construction given in the proof sketch of~\cref{thm:livegroups-product-equivalence}.
    For the rest of the proof, whenever we mention a play on a game, we will assume it starts from the initial vertex of the game.

    We want to show that the initial state $v_\init$ of $\auggame$ is winning if and only if the initial state $(v^1_0, v^2_0)$ of the product game $\product = \game_{\labeling^1} \times G_{\assumplive(E^\ell), \labeling^2}$ is winning. 
    Instead, we will show the equivalence of plays on $\product$ and $\auggame$, i.e. we will show that a play $\rho$ is winning in $\product$ if and only if its projection $\proj^1(\rho)$ to the game $\game_{\labeling^1}$ (and equivalently, to $\auggame$ is winning in $\auggame$.
    
    This argument is sufficient to show the desired result since \begin{inparaenum}[1)] \item there exists a one-to-one correspondance between the plays of $\auggame$ and $\product$, %
    \item for all $n$, from the $n^{th}$ vertex in all plays the same player choses the next state.\label{inpara:two} \end{inparaenum}
    These two conditions together give us the equivalence of underlying game graphs of the two games.
    The latter follows from both games being alternating and the initial vertices belonging to $\po$ in both games.
    
    To see the former, observe that: \begin{inparaenum}[a)] \item all $\pz$ edges are labeled with $a$ in both games,\label{inpara:a} \item if $v_k \in V^1_1$ then $v'_k \in V^2_1$ (follows from~\cref{inpara:two}),\label{inpara:b} \item $u_0$ is the only $\po$ vertex in $G_{\assumplive(E^\ell), \labeling^2}$ and it has outgoing edges labeled with all actions in $\Sigma$.\label{inpara:c}
       \end{inparaenum} 
    \cref{inpara:a},\ref{inpara:b}$\&$\ref{inpara:c} imply that all action sequences that can be taken from $v^1_0$ in $\game_{\labeling^1}$ can also be taken from $v^2_0$ in $G_{\assumplive(E^\ell), \labeling^2}$; and consequently can be taken from $(v^1_0, v^2_0)$ in the product game $\product$.
 
    Furthermore, since any play $\rho$ in $\product$ requires the existence of a play $\proj^1(\rho)$ in $\game_{\labeling^1}$ by the definition of $E_\product$, there is a one-to-one correspondance between the plays of $\product$ and $\auggame$. %

    Now we proceed to prove the equivalence of plays. 
    Let $\rho = (v_0, v'_0) (v_1, v'_1) \ldots$ be a winning play in $\product$. Then by construction, either $\proj^1(\rho)$ satisfies $\spec$ or $\rho$ violates a live group assumption $H_{(u_0, x_i)} = \{ (w, u_0) \to (v, x_i) \in E_\product \mid v,w \in V^1\}$.
    Again by construction, $(w, u_0) \to (v, x_i) \in E_\product$ if and only if there exists an edge $w \xrightarrow{h_i} v$ in $\game_{\labeling^1}$. In $\game_{\labeling^1}$, only the edges in the live group $H_i$ are labeled with $h_i$, 
    so this implies that $(w, v) \in H_i$. Since the live group $H_{(u_0, x_i)}$ is violated in $\rho$, there exists an edge in $(w, u_0) \to (v, x_i) \in H_{(u_0, x_i)}$ enabled infinitely often but non of the edges in $H_{(u_0, x_i)}$ is taken infinitely often.
    This means that in $\proj^1(\rho)$ the edge $w\to v \in H_i$ is enabled infinitely often but non of the edges in $H_i$ is taken infinitely often, which implies that $\proj^1(\rho)$ violates the group liveness condition in $\auggame$. 
    Which gives us that $\proj^1(\rho)$ is winning in $\auggame$.

    Now let $\rho$ be a play in $\product$ that is winning for $\p{1}$. Then by construction $\proj^1(\rho)$ does not satisfy $\spec$ and it satisfies the group liveness assumptions\\ $\livegroups_\product = \cup_{i \in [1;m]}\livegroupSingle_{(u_0, x_i)}$.
    
    Since $\rho$ satisfies the group liveness assumptions, for all $i \in [1;m]$ if an edge $(u_0, w) \to (x_i, v) \in H_{(u_0, x_i)}$ is enabled infinitely often, then an edge in $H_{(u_0, x_i)}$ is taken infinitely often in $\rho$.
    As before, this translates to the projection $\proj^1(\rho)$ as, if an edge $w \to v \in H_i$ is enabled infinitely often, then an edge from $H_i$ is taken infinitely often. This shows that $\proj^1(\rho)$, the projection of $\rho$ to $\auggame$, both violates $\spec$ and satisfies the assumptions, and is thus winning for $\p{1}$. \qed

\end{proof}

\section{Proof of \cref{thm:persGame}}\label{appendix:progress-group}

First, let us remark that one can solve the game $(\gamegraph,\spec)$ with $\spec = \LTLeventually A \vee \LTLalways S$ for some $A,S\subseteq V$ by reducing it to a safety game as formalized below.

\begin{remark}\label{rem:safety}
Given a game $\game = (\gamegraph = ( V, E), \spec)$ where $\spec = \LTLeventually A \vee \LTLalways S$ for some $A,S\subseteq V$, one can reduce the game to a smaller safety game $(\gamegraph', \spec' =\LTLalways S')$, where $S' = S\cup\{v_A\}$ and $\gamegraph'$ is the game graph obtained from $\gamegraph$ by merging all vertices in $A$ to a single new sink vertex $ v_A$, i.e., all incoming edges to $A$ are retained but $ v_A$ has only one outgoing edge that is $( v_A, v_A)$. In such a game, the winning region is $ V\setminus\attro{\gamegraph'}{V\setminus S'}$, see~\cite{gamesbook}.
\end{remark}

Now, let us prove \cref{thm:persGame}.
\restateprogressgroup*
\begin{proof}
    Suppose $\win$ be the winning region in the augmented game $\auggame$. 
    Using induction on the number of times $\reachPers(\cdot)$ is called, we show that the set returned by the algorithm is indeed $\win$, and the updated  strategy $\strat$ returned by the algorithm is a winning strategy in $\auggame$. 
    
    \paragraph*{Base case:}
    If $\reachPers(\cdot)$ is never called, i.e., the algorithm returned $(A,\strat)$ in \cref{alg:reachPers:end}. Hence, we need to show that $A=\win$.
    
    \paragraph*{$(\Rightarrow)$}First, let us show that $A \subseteq \win$.
    By the definition of attractor function $\Attrz{\gamegraph}{T}$, every $\strat_A$-play from $A$ eventually visits $T$, and hence, satisfies $\spec$ (which is stronger than $\assumpPers(\perslivegroups) \Rightarrow \spec$). Therefore, every vertex in $A$ is trivially winning in $\auggame$, and hence, $A \subseteq \win$.
    
    \paragraph*{$(\Leftarrow)$}Now, for the other direction, suppose $ v$ be a vertex such that $ v\not \in A$. It is enough to show that $ v\not \in \win$.
    As $ v\not \in A = \Attrz{\gamegraph}{T}$, $\p{0}$ can not force the plays to visit~$T$. 
    If $v\not\in \perssource$ for every $(\perssource,\persedges,\perstarget)\in \perslivegroups$, then the persistent live group assumptions are not relevant for vertex $ v$. 
    Now, suppose $ v\in\perssource$ for some $(\perssource,\persedges,\perstarget)\in \perslivegroups$.
    As the algorithm did not reach \cref{alg:reachPers:recursion}, for every persistent live group, one of the conditional statements, the one in \cref{alg:reachPers:if} or the one in \cref{alg:reachPers:B}, is not satisfied.
    If the statement in \cref{alg:reachPers:if} is not satisfied, i.e., $(\perssource\setminus A) \cap \pre(A) = \emptyset $, then there is no edge from $\perssource\setminus A$ to $A$, and hence, this persistent live group constraint does not help in reaching $A$ from $ V\setminus A$ anyway. 
    
    Next, if the statement in \cref{alg:reachPers:B} is not satisfied, then it holds that $B \subseteq A$. Hence, $ v\not \in B$. 
    As $B$ is the winning region for game $(\gamegraph|_{\persedges}, \spec_B)$ and such a game is determined~\cite{gamesbook}, $\p{1}$ has a strategy $\strat_1$ such that every $\strat_1$-play in this game starting from $ v$ satisfies $\neg \spec_B = \LTLalways \neg A \wedge \LTLeventually (\perstarget\cup  V\setminus\perssource)$. Therefore, every $\strat_1$-play trivially satisfies $\assumpPers(\perssource,\persedges,\perstarget)$ without ever reaching~$A$. Hence, if $\p{1}$ sticks to strategy $\strat_1$, $\p{0}$ can not make the plays from $ v$ visit $A \supseteq T$ using this constraint. 
    Therefore, in any case, $\p{0}$ has no strategy that can enforce a play from $ v$ to satisfy $\assumpPers(\perslivegroups) \Rightarrow \LTLeventually T$. Hence, $ v\not \in \win$.

    \paragraph*{$(\text{winning strategy})$}Now, let us show that the returned strategy $\strat$ is indeed a winning strategy in $\auggame$. 
    As $\strat_A$ is the attractor strategy to reach $T$, \cref{alg:reachPers:stratA}, it is easy to verify that every $\strat$-play starting from $A\setminus T$ eventually visits $T$, and hence satisfies $\spec$.
    Therefore, every $\strat$-play from $A$ is winning.

    \paragraph*{Induction case:}
    Suppose the algorithm returned $(C,\strat)$ in \cref{alg:reachPers:recursion} for some $(\perssource,\persedges,\perstarget)\in \perslivegroups$. By induction hypothesis, $C$ is the winning region and $\strat_C$ is a  winning strategy in the augmented game $\auggame_C = (\gamegraph,\spec_C,\assumpPers(\perslivegroups))$ with $\spec_C = \LTLeventually  (A\cup B)$. 
    
    \paragraph*{$(\Leftarrow)$}First, let us show that $\win \subseteq C$.
    By the definition of attractor set $\Attrz{\gamegraph}{\cdot}$, it is easy to see that $T\subseteq A$. So, every play in $\gamegraph$ satisfies $\LTLeventually T \Rightarrow \LTLeventually (A\cup B)$. Therefore, a winning play in augmented game $(\gamegraph, T,\assumpPers(\perslivegroups))$ is also winning in augmented game $(\gamegraph, A\cup B,\assumpPers(\perslivegroups))$. Therefore, $\win \subseteq C$.

    \paragraph*{$(\Rightarrow)$}Now, for the other direction, let us first show that $B \subseteq \win$.
    As $\strat_B$ is a winning strategy in game $\auggame_B$, every $\strat_B$-play $\play$ starting in $B$ satisfies $\spec_B$. 
    By definition of $\spec_B$, either $\play$ satisfies $\LTLeventually A$ or it satisfies $\LTLalways (\perssource\setminus\perstarget)$.
    Furthermore, as $\play$ is a play in $\gamegraph|_{\persedges}$, it satisfies 
    $\LTLalways (\perssource\wedge \assumpC(\persedges))$.
    Hence, if $\play$ satisfies $\assumpPers(\perssource,\persedges,\perstarget)$, then it also satisfies $\LTLeventually \perstarget$.
    Therefore, $\play$ can not satisfy both $\assumpPers(\perssource,\persedges,\perstarget)$ and  $\LTLalways (\perssource\setminus\perstarget)$. As a consequence, $\play$ satisfies $\assumpPers(\perssource,\persedges,\perstarget) \Rightarrow \LTLeventually A$.
    Furthermore, as we know, $A\subseteq \win$. 
    Therefore, $\play$ satisfies $\LTLeventually A \Rightarrow \LTLeventually \win$, and hence, satisfies $\assumpPers(\perssource,\persedges,\perstarget) \Rightarrow \LTLeventually \win$. So, every $\strat_B$-play starting in $B$ satisfies $\assumpPers(\perslivegroups) \Rightarrow \LTLeventually \win$. 
    Then, one can construct a $\p{0}$ strategy $\strat_0$ (i.e., the one that uses $\strat_B$ until the play reaches the winning region~$\win$ of game $\auggame$, and then switches to a winning strategy of game $\auggame$) such that every $\strat_0$-play starting in $B$ satisfies the following
    \[(\assumpPers(\perslivegroups)\Rightarrow \LTLeventually \win) \wedge \LTLalways (\win \wedge \assumpPers(\perslivegroups)\Rightarrow \LTLeventually T),\]
    and hence, satisfies $\assumpPers(\perslivegroups) \Rightarrow \LTLeventually T$. Therefore, $B \subseteq \win$.

    Now, let us the other direction for induction case, i.e., $C\subseteq \win$. 
    As $B\subseteq \win$ and $A\subseteq \win$ as proven by the arguments given in base case, 
    it holds that $A\cup B \subseteq \win$.
    So, every play in $\gamegraph$ satisfies $\LTLeventually (A\cup B) \Rightarrow\LTLeventually \win$.
    Furthermore, as $\strat_C$ is a winning strategy in game $\auggame_C$, every $\strat_C$-play starting in $C$ satisfies $\assumpPers(\perslivegroups)\Rightarrow \LTLeventually (A\cup B)$, and hence, satisfies $\assumpPers(\perslivegroups)\Rightarrow \LTLeventually \win$.
    Then, as in the last paragraph, one can construct a $\p{0}$ strategy $\strat_0$ (i.e., the one that uses $\strat_C$ until the play reaches the winning region~$\win$ of game $\auggame$, and then switches to a winning strategy of game $\auggame$) such that every $\strat_0$-play starting in $C$ satisfies the following
    \[(\assumpPers(\perslivegroups)\Rightarrow \LTLeventually \win) \wedge \LTLalways (\win \wedge \assumpPers(\perslivegroups)\Rightarrow \LTLeventually T).\]
    Hence, every $\strat_0$-play starting in $C$ satisfies $\assumpPers(\perslivegroups)\Rightarrow \LTLeventually T$. Therefore, $C \subseteq \win$.

    \paragraph*{$(\text{winning strategy})$}Now, let us show that the returned strategy $\strat$ in \cref{alg:reachPers:recursion} is also a winning strategy in game $\auggame$.
    As $\strat$ is follows strategy $\strat_C$ for vertices in  $C\setminus (A\cup B)$, every $\strat$-play from $C\setminus (A\cup B)$ eventually visits $A\cup B$ when $\assumpPers(\perslivegroups)$ holds.
    Now, let $\strat_M$ be the updated strategy until \cref{alg:reachPers:stratB}.
    Then, from \cref{alg:reachPers:stratA},\ref{alg:reachPers:stratB}, it is easy to see that $\strat(v) = \strat_M(v)$ for every vertex $v$ in $A\cup B$.
    As $\strat_B$ is a winning strategy in game $\auggame_B$, using \cref{alg:reachPers:stratB} and the discussion above, every $\strat$-play from $B\setminus A$ eventually visits $A$ when $\assumpPers(\perslivegroups)$ holds.
    Then, using arguments of base case, every $\strat$-play from $A\setminus T$ eventually visits $T$.
    Therefore, in total, every $\strat$-play from $C$ eventually visits $T$ when $\assumpPers(\perslivegroups)$ holds.
    Hence, $\strat$ is indeed a winning strategy in game $\auggame$.
    
    \paragraph*{Time complexity:}
    Let $k$ be the number of times $\reachPers(\cdot)$ is called. If $T= V$, then $A =  V$, and hence, $\perssource\setminus A = \emptyset$ for every $(\perssource,\persedges,\perstarget)\in \perslivegroups$, and hence, $\reachPers(\cdot)$ will never be called. Furthermore, if $T\neq  V$, then, by definition of $\Attrz{\gamegraph}{\cdot}$, it holds that $T\subseteq A$. So, in \cref{alg:reachPers:if}, we keep adding at least one vertex to the target for the next call of $\reachPers(\cdot)$. Hence, $k$ can be at most $\abs{ V}$.
    Moreover, in each iteration, we might need to solve game $(\gamegraph|_{\persedges}, \spec_B)$ for each $(\perssource,\persedges,\perstarget)\in \perslivegroups$; and using \cref{rem:safety}, solving such a game can be reduced to computing an attractor function $\attro{\gamegraph}{\cdot}$. As computing such an attractor function takes
    $\mathcal{O}(\abs{ E})$ time~\cite{gamesbook}, the algorithm takes $\mathcal{O}(\abs{\perslivegroups}\cdot\abs{ V}\cdot\abs{ E})$ time in total.\qed
\end{proof}

\section{Quasi-polynomial Algorithm for Parity Games Augmented with Persistent Live Groups}\label{appendix:quasi-progress-groups}\label{appendix:lemma-prog-restricted}
Let us first prove the lemma showing defined restriction to persistent live groups for restricted game graph is equivalent to original one.
\restateprogressgrouplemma*
\begin{proof}
    If $\play$ is winning in $\auggame$, then either it satisfies the parity objective $\spec$ or doesn't satisfy the persistent live group assumptions $\assumpPers(\perslivegroups)$. 
    If it satisfies $\spec$, then it also satisfies $\spec|_U$ as it only visits vertices in $U$.
    Else if it doesn't satisfy persistent live group assumptions, then there exists a persistent live group $(\perssource,\persedges,\perstarget)$ such that $\play$ doesn't satisfy $\assumpPers(\perssource,\persedges,\perstarget)$.
    Hence, there exists a suffix $\play'$ of $\play$ that satisfies $\LTLalways(\perssource\wedge\assumpC(\persedges))\wedge \LTLalways\neg \perstarget$.
    As $\play'$ stays inside $U$ and edges of $\persedges$ used in $\play'$ also belong to $\persedges|_U$, $\play'$ also satisfies $\LTLalways(\perssource|_U\wedge \assumpC(\persedges|_U))\wedge\LTLalways\neg \perstarget|_U$.
    Hence, $\play$ also doesn't satisfy the persistent live group assumptions $\assumpPers(\perslivegroups|_U)$.
    Therefore, in all cases, $\play$ is also winning in $\auggame|_U$.
    
    Now, for the other direction, let $\play$ is not winning in $\auggame$. 
    Hence, $\play$ satisfies the persistent live group assumptions $\assumpPers(\perslivegroups)$ but not parity objective $\spec$. 
    So, $\play$ also doesn't satisfy parity objective $\spec|_U$.
    We only need to show that $\play$ satisfies $\assumpPers(\perslivegroups|_U)$.
    Consider a persistent live group $(\perssource,\persedges,\perstarget)\in\perslivegroups$ and a suffix $\play'$ of $\play$.
    It is enough to show that $\play'$ satisfies $\assumpPers(\perssource|_U,\persedges|_U,\perstarget|_U)$.
    Since the suffix $\play'$ satisfies $\assumpPers(\perssource,\persedges,\perstarget)$, it satisfies one of the following: $\LTLeventually T$ or $\LTLeventually\neg\perssource$ or $\LTLeventually\neg\assumpC(\persedges)$.
    If $\play'$ satisfies $\LTLeventually T$ or $\LTLeventually\neg\perssource$, as $\play'$ always stays inside $U$, it also satisfies $\LTLeventually T|_U$ or $\LTLeventually\neg\perssource|_U$, respectively and hence, it satisfies $\assumpPers(\perssource|_U,\persedges|_U,\perstarget|_U)$.
    Else if $\play'$ satisfies $\LTLeventually\neg\assumpC(\persedges)$, then for some vertex $v\in\source(\persedges)$, $\play'$ visits $v$ but doesn't take any edge in $\persedges$ from $v$.
    If $v\in\source(\persedges|_U)$, then $\play'$ also satisfies $\LTLeventually\neg\assumpC(\persedges|_U)$.
    Otherwise, if $v\not\in\source(\persedges|_U)$, then by definition, $v\not\in\perssource|_U$, and hence, $\play'$ satisfies $\LTLeventually\neg\perssource|_U$.
    Therefore, in all cases, $\play'$ satisfies $\assumpPers(\perssource|_U,\persedges|_U,\perstarget|_U)$.
    So, $\play$ is not winning in $\auggame|_U$.\qed
\end{proof}

\input{sections/quasi-progress-groups.tex}

%% file: sections/CNF-edges.tex
\begin{definition}\label{def:live-CNF}
    Given a game graph $\gamegraph = (V,E)$, a \emph{\CNFname} $\varphi_v$ for some vertex $v$ is defined as the formula
    \[\varphi_v = \bigwedge_{i\in[1;m]}\bigvee_{j\in[1;n_i]} e_{ij},\]
    where $e_{ij}\in E(v)$ for all $i\in[1;m],j\in[1;n_i]$.
    The assumption represented by a set $\liveCNF = \{\varphi_v\mid v\in V\}$ of \CNFnames is expressed by the following LTL formula
    \begin{equation}\label{eq:assumpliveCNF}
        \assumpCNFlive(\liveCNF) \coloneqq \bigwedge_{v\in V} \LTLalways\LTLeventually v \Rightarrow \LTLalways\LTLeventually \varphi_v.
    \end{equation}
    We write games augmented with \CNFnames $\liveCNF$ to refer to the augmented games $\auggame = (\gamegraph,\spec,\assumpCNFlive(\liveCNF))$.
\end{definition}

One can see that such assumptions can actually be represented by a conjunctions of multiple live groups for each vertex.
So, one can generalize the idea illustrated in \cref{fig:boolean-live-groups} to get similar time complexity results for games augmented with \CNFnames as formalized below.
\begin{theorem}\label{thm:live-CNF}
    Reachability (resp. parity) games augmented with \CNFnames assumptions can be solved in $\PTIME$ (resp. $\QP$).
\end{theorem}
\begin{proof}
    Consider an augmented game $\auggame = (\gamegraph = (V,E),\spec,\assumpCNFlive(\liveCNF))$ with $\liveCNF = \{\varphi_v \mid v\in V\}$.
    Consider some formula $\varphi_v = \bigwedge_{i\in[1;m]}\vee_{j\in[1;n_i]} e_{ij}$ in $\liveCNF$.
    Then, we can re-write $\LTLalways\LTLeventually \varphi_v$ as
    follows
    \begin{align*}
        \LTLalways\LTLeventually \varphi_v &= \bigwedge_{i\in[1;m]}\LTLalways\LTLeventually\bigvee_{j\in[1;n_i]} e_{ij}\\
        &= \bigwedge_{i\in[1;m]}\LTLalways\LTLeventually \livegroupSingle_i \quad \text{with } \livegroupSingle_i = \{e_{ij}\mid j\in [1;n_i]\}
    \end{align*}
    As $\source(\livegroupSingle_i) = \{v\}$ for each $i\in[1;m]$, by definition, $\LTLalways\LTLeventually v \Rightarrow \LTLalways\LTLeventually \varphi_v$ is equivalent to the live group assumption $\assumpgrlive(\livegroups_v)$ with $\livegroups_v = \{H_i\mid i\in[1;m]\}$.
    Hence, for each vertex $v$, there exists a set $\livegroups_v$ of live groups such that conjunctions of their assumptions represents the \CNFnames assumption, i.e.,
    \[\assumpCNFlive(\liveCNF) = \bigwedge_{v\in V}\assumpgrlive(\livegroups_v) = \assumpgrlive\left(\bigcup_{v\in V}\livegroups_v\right).\]
    Now, using the idea discussed above, we construct an equivalent game augmented with live edges.
    Let $\auggame' = (\gamegraph' = (V',E'),\spec',\assumplive(\livedges))$ be an augmented defined as follows:
    \begin{itemize}
        \item $V_0' = V_0$
        \item $V_1' = V_1\cup \{v_\livegroupSingle \mid \livegroupSingle\in \livegroups_v \text{ for some }v\in V\}$
        \item $E' = E \cup \{(u,v_\livegroupSingle) \mid \livegroupSingle\in\livegroups_v\} \cup \{(v_\livegroupSingle,w)\mid (u,w)\in\livegroupSingle\}$
        \item $\livedges = \{(u,v_\livegroupSingle)\mid \livegroupSingle\in\livegroups_v\}$
        \item for reachability objective, $\spec' = \spec$, and for parity objectives, $\spec'$ uses same priority function as $\spec$ with priority $0$ for all new vertices.
    \end{itemize} 
    
    Notice that in both augmented games $\auggame$ and $\auggame'$, $\pz$ has a positional winning strategy.
    That is because, as mentioned before, the proof of \cref{lem:NPinclusion-augmentedparity} shows that every parity (and thus, reachability) game augmented with live groups can be viewed as a Rabin game, and such property holds for Rabin games~\cite{EJ99}.
    As each live edge is a singleton live group, the same result also holds for games augmented with live edges.
    Hence, we only need to consider positional strategies for $\pz$ in both games.

    Furthermore, as the vertex set of $\pz$ and her edges are same in both games, every positional strategy of $\pz$ in $\auggame$ also acts as a positional strategy of $\pz$ in $\auggame'$ and vice versa.

    Moreover, as $\auggame'$ is a game augmented with live edges, by \cref{thm:live-edges}, one can solve it polynomial (resp. quasi-polynomial) time for reachability (resp. parity) objective.
    Hence, it is enough to show the equivalence between two games in terms of positional strategies of $\pz$ as formalized in the following claim.

    \begin{claim}
        For every vertex $v\in V$, a positional $\pz$ strategy is winning from $v$ in $\auggame$ if and and only if it is winning from $v$ in $\auggame'$.
    \end{claim}
    \paragraph*{$(\Rightarrow)$} Suppose $\strat$ be a positional strategy for $\pz$ that is winning from some vertex $v\in V$ in game $\auggame$.
    Let $\play'$ be a $\strat$-play from $v$ in $\auggame'$.
    We need to show that $\play'$ is winning in $\auggame'$.
    First of all, by the construction, every edge $(u,w)$ in $\auggame$ is either retained in $\auggame'$ or $u$ is connected to $w$ via some new vertex $v_\livegroupSingle$.
    Hence, removing all new vertices (i.e., of the form $v_\livegroupSingle$) from $\play'$ gives us the corresponding $\strat$-play $\play$ in $\auggame$.
    As $\strat$ is winning from $v$, so is the play $\play$.
    If $\play$ satisfies the specification $\spec$, then by construction, it is easy to see that $\play'$ also satisfies the correspoding specification $\spec'$ and hence, is winning.
    Suppose $\play$ doesn't satisfy $\spec$, then it violates the assumption $\assumpgrlive(\livegroups_u)$ for some $u\in V$.
    Hence, $\play$ visits $u$ infinitely often but the edges in some live group $\livegroupSingle\in\livegroups_v$ appears only finitely often in $\play$.
    By construction, $\play'$ also visits $u$ infinitely often and there is an edge $(u,v_\livegroupSingle)$ in $\auggame'$. 
    Furthermore, the edges from $v_\livegroupSingle$ leads to the set $\livegroupSingle(u)$.
    Hence, if edge $(u,v_\livegroupSingle)$ appears in $\play'$ infinitely often, then there exists a vertex $w\in\livegroupSingle(u)$ such that the edge $(v_\livegroupSingle,w)$ appears infinitely often.
    Then, by construction of $\play$, edge $(u,w)$ appears infinitely often, which is a contradiction as $w\in\livegroupSingle(u)$.
    Therefore, $\play'$ visits $u$ infinitely often but the live edge $(u,v_\livegroupSingle)$ only appears finitely often.
    So, $\play'$ violates the live edge assumption, and hence, is winning.

    \paragraph*{$(\Leftarrow)$} For the other direction, let $\strat$ be a positional strategy for $\pz$ that is winning from some vertex $v\in V$ in game $\auggame'$ and let $\play$ be a $\strat$-play in $\auggame$.
    We need to show that $\play$ is winning in $\auggame$.
    By construction, there exists $\strat$-plays in $\auggame'$ such that restricting them to $V$ results in play $\play$.
    Consider one of such play $\play'$ such that for each edge $(u,w)\in \livegroupSingle$ in some live group appearing infinitely often, $\play'$ visits $v_\livegroupSingle$ (from $u$) infinitely often
    (if such an edge $(u,w)$ belongs to multiple live groups $\livegroupSingle$, then $\play'$ visits all their corresponding vertices $v_\livegroupSingle$ infinitely often by alternating between them).
    As $\play'$ is a $\strat$-play in $\auggame'$, either it satisfies the specification $\spec'$ or it violates some live edge assumption $(u,v_\livegroupSingle)$.
    For the former case, as argued before, $\play$ also satisfies the specification $\spec$ and hence, is winning.
    For the later case, $\play'$ visits $u$ infinitely often but visits $v_\livegroupSingle$ finitely often.
    By construction of $\play'$, if some edge $(u,w)\in\livegroupSingle$ appears infinitely often in $\play$, then $\play'$ visits $v_\livegroupSingle$ infinitely often, which is a contradiction.
    Hence, no edges in $\livegroupSingle$ appears infinitely often in $\play$.
    However, as $\play'$ visits $u$ infinitely often, so does $\play$.
    Therefore, $\play$ violates the live group assumption $\assumpgrlive(\livegroups_u)$, and hence, is winning.\qed
\end{proof}

%% file: sections/quasi-progress-groups.tex
\algblockdefx[PROCEDURE]{Procedure}{EndProcedure}[3][]{\textbf{procedure} \textsc{#2}(#3)\ifthenelse{\equal{#1}{}}{}{\Comment{#1}}\\\textbf{begin}}{\textbf{end}}
\algloopdefx{If}[1]{\textbf{if} #1 \textbf{then}}
\algblockdefx[WHILE]{While}{EndWhile}[1]{\textbf{while} #1 \textbf{do begin}}{\textbf{end}}
\algblockdefx[DOWHILE]{DoBegin}{EndDoWhile}{\textbf{do begin}}[1]{\textbf{end while} #1}

Now, we provide a quasi-polynomial algorithm in \cref{alg:quasi-algo-progress} for parity games augmented with persistent live groups.
This algorithm is obtained by replacing every use of $\Attri{\gamegraph}{T}$ with $\AtrI(\gamegraph,T,\perslivegroups)$ in the algorithm given by Parys~\cite{quasiZielonka}.
Here, $\AtrE(\gamegraph,T,\perslivegroups)$ returns the winning region for the augmented reachability game $(\gamegraph,\LTLeventually T,\assumpPers(\perslivegroups))$, which can be computed in polynomial time by the procedure $\reachPers(\gamegraph,T,\perslivegroups)$.
Furthermore, $\AtrO(\gamegraph,T,\perslivegroups)$ is simply the attractor set $\attro{\gamegraph}{T}$ for $\po$.

We also provide the complexity analysis and the proof of correctness for the algorithm which almost same as the proof given by Parys~\cite{quasiZielonka} with some modifications.
\cref{theorem:pers-solve_augmented_parity} follows from the correctness of the algorithm along with its quasi-polynomial time complexity.

\begin{algorithm}[b!]
	\caption{$\solvePers(\gamegraph,\parity(\priority),\perslivegroups)$}\label{alg:quasi-algo-progress}
	\begin{algorithmic}[1]
        \Require An augmented parity game $\auggame = (\gamegraph,\parity(\priority),\assumpPers(\perslivegroups))$ with game graph $\gamegraph = (V,E)$, some persistent live groups $\perslivegroups$, and priority function $\priority: V\rightarrow [0;2d]$ 
        \Ensure Winning region and memoryless winning strategy in the augmented game $\auggame$  
    \Return $\qsolve_0(\auggame,2d,\abs{V},\abs{V})$
    \Procedure[$p_0,p_1$ are ``precision'' parameters]{QSolve$_i$}{$\auggame, h, p_i, p_{1-i}$}
		\If {$\auggame=\emptyset\lor p_i\leq 1$} \State \Return $\emptyset$;\Comment{w.l.o.g. assume that there are no self-loops in $\auggame$}
		\DoBegin
			\State $N_h\gets\{v\in\ V\mid \priority(v)=h\}$;
			\State $\Ha\gets \auggame|_{V\setminus\AtrI(G,N_h,\perslivegroups)}$;
			\State $W_{1-i}\gets \textsc{QSolve}_{1-i}(\Ha,h-1,\lfloor p_{1-i}/2\rfloor,p_i)$;
			\State $\auggame \gets \auggame|_{V\setminus\Atr{1-i}(G,W_{1-i},\perslivegroups)}$
		\EndDoWhile{$W_{1-i}\neq\emptyset$}
        \State $N_h\gets\{v\in\ V\mid \priority(v)=h\}$;
        \State $\Ha\gets \auggame|_{V\setminus\AtrI(G,N_h,\perslivegroups)}$;
        \State $W_{1-i}\gets \textsc{QSolve}_{1-i}(\Ha,h-1,p_{1-i},p_i)$;
        \State $\auggame \gets \auggame|_{V\setminus\Atr{1-i}(G,W_{1-i},\perslivegroups)}$
		\While{$W_1\neq\emptyset$}
        \State $N_h\gets\{v\in\ V\mid \priority(v)=h\}$;
        \State $\Ha\gets \auggame|_{V\setminus\AtrI(G,N_h,\perslivegroups)}$;
        \State $W_{1-i}\gets \textsc{QSolve}_{1-i}(\Ha,h-1,\lfloor p_{1-i}/2\rfloor,p_i)$;
        \State $\auggame \gets \auggame|_{V\setminus\Atr{1-i}(G,W_{1-i},\perslivegroups)}$
		\EndWhile
		\State \Return $\ V$;
	\EndProcedure
	\end{algorithmic}
	\end{algorithm}

\subsection*{Complexity Analysis}

Let us analyze the complexity of our algorithm.
Let $R(h,l)$ be the number of (nontrivial) executions of the $\textsc{QSolve}_0$ and $\textsc{QSolve}_1$ procedures performed during one call to $\textsc{QSolve}_0(\auggame,h,p_0,p_1)$ with 
$\lfloor\log p_0\rfloor+\lfloor\log p_1\rfloor=l$, 
and with $\auggame$ having at most $n$ vertices (where $n$ is fixed).
We only count here nontrivial executions, that is, such that do not leave the procedure in line 4.
Clearly $R(0,l)=R(h,0)=0$.
For $h,l\geq 1$ it holds that
\begin{align}
	R(h,l)\leq 1+n\cdot R(h-1,l-1)+R(h-1,l)\,.\label{eq:1}
\end{align}
Indeed, in $\textsc{QSolve}_0$ after every call to $\textsc{QSolve}_1$ we remove at least one vertex from $\auggame$,
with the exception of two such calls: the last call in line 8, and the last call ever.
In effect, in lines 8 and 18 we have at most $n$ calls to $\textsc{QSolve}_1$ with decreased precision 
(plus, potentially, the $(n+1)$-th call with empty $\auggame$, which is not included in $R(h,l)$),
and in line 13 we have one call to $\textsc{QSolve}_1$ with full precision. 
Notice that $\lfloor\log p_1\rfloor$ (hence also $l$) decreases by $1$ in the decreased-precision call.

Using Inequality~\eqref{eq:1} we now prove by induction that $R(h,l)\leq n^l\cdot\binom{h+l}{l}-1$.
For $h=0$ and for $l=0$ the inequality holds.
For $h,l\geq 1$ we have that 
\begin{align*}
	R(h,l)&\leq 1+n\cdot R(h-1,l-1)+R(h-1,l)\\
	&\leq 1+n\cdot\left(n^{l-1}\cdot\binom{h-1+l-1}{l-1}-1\right)+n^l\cdot\binom{h-1+l}{l}-1\\
	&\leq n^l\cdot\left(\binom{h-1+l}{l-1}+\binom{h-1+l}{l}\right)-1\\
	&=n^l\cdot\binom{h+l}{l}-1\,.
\end{align*}
In effect, $R(h,l)\leq n^l\cdot(h+l)^l$.
Recalling that we start with $l=2\cdot\lfloor\log n\rfloor$, we see that this number is quasi-polynomial in $n$ and $h$.
This concludes the proof, 
since obviously a single execution of the $\textsc{QSolve}_0$ procedure (not counting the running time of recursive calls) costs polynomial time.

\subsection*{Correctness}

We now justify correctness of the algorithm which is mainly based on properties of \emph{dominions}. Given a game $(\gamegraph,\spec)$, a set $U\subseteq V$ is said to be a \emph{dominion} for $\p{i}$ if they can enforce every play from $U$ to stay inside $U$ in addition to being winning for them, i.e., there exists a $\p{i}$ strategy $\strat$ such that every $\strat$-play from $U$ satisfies $\LTLalways U$ and is winning for $\p{i}$.
Dominions for augmented games $(\gamegraph,\spec,\assump)$ are defined w.r.t. the equivalent game $(\gamegraph,\assump\Rightarrow\spec)$.

Now, it is easy to see that winning region $\wini$ of $\p{i}$ is the largest dominion for that player. Using this information, the correctness of the algorithm amounts to proving the following theorem.

\begin{theorem}\label{theorem:quasi-progress-groups}
	Procedure $\textsc{QSolve}_0(\auggame,h,p_0,p_1)$ returns a set $W_0$ such that for every $S\subseteq\ V$,
	\begin{itemize}
	\item	if $S$ is a dominion for $\p{0}$, and $|S|\leq p_0$, then $S\subseteq W_0$, and
	\item	if $S$ is a dominion for $\p{1}$, and $|S|\leq p_1$, then $S\cap W_0=\emptyset$.
	\end{itemize}
\end{theorem}

Notice that in $\auggame$ there may be vertices that do not belong to any dominion smaller than $p_0$ or $p_1$; 
for such vertices we do not specify whether or not they are contained in $W_0$.

Before proving \cref{theorem:quasi-progress-groups}, let us observe two important facts about dominions 
that are important for the proof.

\begin{lemma}\label{fact:usun-moj}
	If $S$ is a dominion for $\p{P}$ in an augmented parity game $\auggame = (G=(V,E),\spec,\assumpPers(\perslivegroups))$, and $X$ is a set of vertices of $\auggame$, then $S\setminus\AtrP(G,X,\perslivegroups)$ is a dominion for $\p{P}$ in augmented game $\auggame|_{V\setminus\AtrP(G,X,\perslivegroups)}$.
\end{lemma}

\begin{proof}
	Denote $S'=S\setminus\AtrP(G,X,\perslivegroups)$ and $\auggame'=\auggame|_{V\setminus\AtrP(G,X,\perslivegroups)}$.
	By definition, $\p{P}$ has a strategy $\strat$ such that every $\strat$-play from from $S$ always stays in $S$ and is winning for that player  w.r.t. $\auggame$.
	The strategy $\strat$ remains valid in $\auggame'$, because every vertex $u$ of $\p{P}$ that remains in $\auggame'$ has the same successors in $\auggame'$ as in $\auggame$
	(conversely: if some of successors of $u$ belongs to $\AtrP(G,X,\perslivegroups)$, then $u$ also belongs to $\AtrP(G,X,\perslivegroups)$).	
	Hence, every $\strat$-play from $S'$ in $\auggame'$ always stays in $S'$ and is winning for $\p{P}$ w.r.t. $\auggame$, which implies it is also winning $\p{P}$ w.r.t. $\auggame'$ by \cref{lemma:prog-restricted}. 
	Therefore, $S'$ is also a dominion for $\p{P}$ in~$\auggame'$.
\qed
\end{proof}

\begin{lemma}\label{fact:usun-jego}
	If $S$ is a dominion for $\p{P}$ in an augmented parity game $\auggame = (G=(V,E),\spec,\assumpPers(\perslivegroups))$, and $X$ is a set of vertices of $\auggame$ such that $S\cap X=\emptyset$, 
	then $S$ is a dominion for $\p{P}$ in augmented game $\auggame|_{V\setminus\AtrPbar(G,X,\perslivegroups)}$ (in particular $S\subseteq V\setminus\AtrPbar(G,X,\perslivegroups)$).
\end{lemma}

\begin{proof}
	Denote $\auggame'=\auggame|_{V\setminus\AtrPbar(G,X,\perslivegroups)}$.
	Suppose that there is some vertex $v\in S\cap\AtrPbar(G,X,\perslivegroups)$.
	On the one hand, $\p{P}$ can guarantee that, while starting from $v$, the play stays in $S$ (by the definition of a dominion);
	on the other hand, $\p{1-P}$ can force to reach the set $X$ (by the definition of an attractor and \cref{thm:persGame}), which is disjoint from $S$.
	Thus such a vertex $v$ could not exist, we have $S\subseteq V\setminus\AtrPbar(G,X,\perslivegroups)$.
	
	Moreover, as $S$ remains unchanged in $\auggame'$,
	it is easy to see that, $\p{P}$ has a strategy $\strat$ in $\auggame'$ such that every $\strat$-play from $S$ stays is inside $S$ and is winning for that player.
\qed
\end{proof}

We are now ready to prove \cref{theorem:quasi-progress-groups}.

\begin{proof}[of \cref{theorem:quasi-progress-groups}]
	We prove the theorem by induction on $h$.
	Consider some execution of the procedure.
	By $\auggame^i, N_h^i, \Ha^i, W_1^i$, we denote values of the variables $\auggame, N_h, \Ha, W_1$ just after the $i$-th call to $\textsc{QSolve}_1$ in one of the lines 8, 13, 18;
	in lines 9, 14, 19 we create $\auggame^{i+1}$ out of $\auggame^i$ and $W_1^i$.
	Furthermore, by $\gamegraph^i,V^i,E^i,\spec^i,\perslivegroups^i,\assump^i$, we denote the corresponding game graph, vertex set, edge set, specification, set of persistent live groups, persistent live group assumptions, respectively of game $\auggame^i$.
	In particular $\auggame^1$ equals the original game $\auggame$, and at the end we return $V^{m+1}$, where $m$ is the number of calls to $\textsc{QSolve}_1$.

	Concentrate on the first item of the theorem: fix an $\p{0}$'s dominion $S$ in $\auggame$ (i.e., in $\auggame^1$) such that $|S|\leq p_0$.
	Assume that $S\neq\emptyset$ (for $S=\emptyset$ there is nothing to prove).
	Notice first that a nonempty dominion has at least two vertices (by assumption there are no self-loops in $\auggame$, hence every play has to visit at least two vertices),
	thus, because $S\subseteq\ V$ and $|S|\leq p_0$, we have that $\auggame\neq\emptyset$ and $p_0>1$.
	It means that the procedure does not return in line 4.
	We thus need to prove that $S\subseteq V^{m+1}$.
	
	We actually prove that $S$ is a dominion for $\p{0}$ in $\auggame^i$ for every $i\in [1;m+1]$, meaning in particular that $S\subseteq V^i$.
	This is shown by an internal induction on $i$.
	The base case ($i=1$) holds by assumption.
	For the induction step, consider some $i\in [1;m]$.
	By the induction assumption $S$ is a dominion for $\p{0}$ in $\auggame^i$, and we need to prove that it is a dominion for $\p{0}$ in $\auggame^{i+1}$.
	
	Consider $S^i=S\cap (V^i\setminus \AtrE(\gamegraph^i,N_h^i,\perslivegroups^i))$.
	Because $S^i=S\setminus\AtrE(G^i,N_h^i,\perslivegroups^i)$, by \cref{fact:usun-moj} the set $S^i$ is a dominion for $\p{0}$ in $\Ha^i=\auggame^i|_{V^i\setminus\AtrE(G^i,N_h^i,\perslivegroups^i)}$, and obviously $|S^i|\leq|S|\leq p_0$.
	By the assumption of the external induction (which can be applied to $\textsc{QSolve}_1$, by symmetry) 
	it follows that $S^i\cap W_1^i=\emptyset$, so also $S\cap W_1^i=\emptyset$ (because $W_1^i$ contains only vertices of $\auggame^i$, while $S\setminus S^i$ contains no vertices of $\auggame^i$).
	Thus, by \cref{fact:usun-jego} the set $S$ is a dominion for $\p{0}$ in $\auggame^{i+1}=\auggame^i|_{V^i\setminus\AtrO(G^i,W_1^i,\perslivegroups^i)}$.
	This finishes the proof of the first item.

	Now we prove the second item of the theorem.
	To this end, fix some $\p{1}$'s dominion $S$ in $\auggame$ such that $|S|\leq p_1$.
	If $p_0\leq 1$, we return $W_0=\emptyset$ (line 4), so clearly $S\cap W_0=\emptyset$.
	The interesting case is when $p_0\geq 2$.
	Denote $S^i=S\cap V^i$ for all $i\in [1;m+1]$; we first prove that $S^i$ is a dominion for $\p{1}$ in $\auggame^i$.
	This is shown by induction on $i$.
	The base case of $i=1$ holds by assumption, because $\auggame^1=\auggame$ and $S^1=S$.
	For the induction step, assume that $S^i$ is a dominion for $\p{1}$ in $\auggame^i$, for some $i\in [1;m]$.
	By definition $\auggame^{i+1}=\auggame^i|_{V^i\setminus\AtrO(G^i,W_1^i,\perslivegroups^i)}$ and $S^{i+1}=S^i\setminus\AtrO(G^i,W_1^i,\perslivegroups^i)$,
	so $S^{i+1}$ is a dominion for $\p{1}$ in $\auggame^{i+1}$ by \cref{fact:usun-moj}, which finishes the inductive proof.
	
	For $i\in [1;m]$, let $Z^i$ be the set of vertices (in $S^i\setminus N_h^i$) from which $\p{1}$ has a strategy $\strat^i$ such that every $\strat^i$-play from $Z^i$ satifies $\assump^i\wedge\neg\spec^i\wedge\LTLalways (S^i\setminus N_h^i)$ in $\auggame^i$
	(that is, where $\p{1}$ can win without seeing priority $h$---the highest even priority).
	Let us observe that if $S^i\neq\emptyset$ then $Z^i\neq\emptyset$ ($\clubsuit$).
	Indeed, suppose to the contrary that $Z^i=\emptyset$, 
	and consider an $\p{1}$'s strategy $\strat^i_2$ such that every $\strat^i_2$-play from some vertex $v_0\in S^i$  satifies $\assump^i\wedge\neg\spec^i\wedge\LTLalways S^i$ in $\auggame^i$.
	Because $v_0\not\in Z^i$, this strategy cannot ensure $\LTLalways (S^i\setminus N_h^i)$, 
	so $\p{0}$, while playing against this strategy, can reach a vertex $v_1$ in $N_h^i$ (as she cannot violate the parity condition nor leave $S^i$).
	For the same reason, because $v_1\not\in Z^i$, $\p{0}$ can continue and reach a vertex $v_2$ in $N_h^i$.
	Repeating this forever, $\p{0}$ gets priority $h$ (which is even and is the highest priority) infinitely many times, resulting in a $\strat^i_2$-play that satisfies $\spec^i$ and hence, contradicting the assumptions about the strategy~$\strat^i_2$.
	
	Observe also that from vertices of $Z^i$, $\p{1}$ can actually satisfy $\assump^i\wedge\neg\spec^i\wedge\LTLalways Z^i$ in $\auggame^i$,
	using the strategy $\strat^i$ that allows him to satisfy $\assump^i\wedge\neg\spec^i\wedge\LTLalways (S^i\setminus N_h^i)$.
	Indeed, if a $\strat^i$-play $\play$ enters some vertex $v\in Z^i$, then from this vertex $v$, $\p{1}$ can still satisfy $\assump^i\wedge\neg\spec^i\wedge\LTLalways (S^i\setminus N_h^i)$,
	which means that these vertices belongs to $Z^i$.
	It follows that $Z^i$ is a dominion for $\p{1}$ in $\auggame^i$.
	Moreover, because $Z^i\cap N_h^i=\emptyset$, from \cref{fact:usun-jego} we have that $Z^i$ is a dominion for $\p{1}$ in $\Ha^i=\auggame^i|_{V^i\setminus\AtrE(G^i,N_h^i,\perslivegroups^i)}$.
	
	Let $k$ be the number of the call to $\textsc{QSolve}_1$ that is performed in line 13 (calls number $1,\dots,k-1$ are performed in line 8, and calls number $k+1,\dots,m$ are performed in line 18).
	Recall that $W_1^i$ is the set returned by a call to $\textsc{QSolve}_1(\Ha^i,h-1,p_1^i,p_0)$, where $p_1^k=p_1$, and $p_1^i=\lfloor\frac{p_1}{2}\rfloor$ if $i\neq k$.
	From the assumption of the external induction, if $|Z^i|\leq\lfloor\frac{p_1}{2}\rfloor$ or if $i=k$ (since $Z^i\subseteq S^i\subseteq S$ and $|S|\leq p_1$, clearly $|Z^i|\leq p_1$), 
	we obtain that $Z^i\subseteq W_1^i$ ($\spadesuit$).
	
	We now prove that $|S^{k+1}|\leq\lfloor\frac{p_1}{2}\rfloor$.
	This clearly holds if $S^{k-1}=\emptyset$, because $S^{k+1}\subseteq S^k\subseteq S^{k-1}$.
	Suppose thus that $S^{k-1}\neq\emptyset$.
	Then $Z^{k-1}\neq\emptyset$, by ($\clubsuit$).
	On the other hand, $W_1^{k-1}=\emptyset$, because we are just about to leave the loop in lines 5-10 (the $k$-th call to $\textsc{QSolve}_1$ is in line 13).
	By ($\spadesuit$), if $|Z^{k-1}|\leq\lfloor\frac{p_1}{2}\rfloor$, then $Z^{k-1}\subseteq W_1^{k-1}$, which does not hold in our case.
	Thus $|Z^{k-1}|>\lfloor\frac{p_1}{2}\rfloor$.
	Because $W_1^{k-1}=\emptyset$, we simply have $\auggame^k=\auggame^{k-1}$, and $S^k=S^{k-1}$, and $Z^k=Z^{k-1}$.
	Using ($\spadesuit$) for $i=k$, we obtain that $Z^k\subseteq W_1^k$, and because $S^{k+1}=S^k\setminus\AtrO(G^k,W_1^k,\perslivegroups^k)\subseteq S^k\setminus W_1^k\subseteq S^k\setminus Z^k$
	we obtain that $|S^{k+1}|\leq|S^k|-|Z^k|\leq p_1-(\lfloor\frac{p_1}{2}\rfloor+1)\leq\lfloor\frac{p_1}{2}\rfloor$, as initially claimed.
	
	If $k=m$, we have $Z^m\subseteq W_1^m$ by ($\spadesuit$).
	If $k+1\leq m$, we have $S^m\subseteq S^{k+1}$ (our procedure only removes vertices from the game) and $Z^m\subseteq S^m$, 
	so $|Z^m|\leq\lfloor\frac{p_1}{2}\rfloor$ by the above paragraph,
	and also $Z^m\subseteq W_1^m$ by ($\spadesuit$).
	Because after the $m$-th call to $\textsc{QSolve}_1$ the procedure ends, we have $W_1^m=\emptyset$, so also $Z^m=\emptyset$, and thus $S^m=\emptyset$ by ($\clubsuit$).
	We have $S^{m+1}\subseteq S^m$, so $S^{m+1}=S\cap V^{m+1}=\emptyset$.
	This is exactly the conclusion of the theorem, since the set returned by the procedure is $V^{m+1}$.\qed
\end{proof}